\documentclass[12pt]{article}
\usepackage[english]{babel}
\usepackage[utf8]{inputenc}
\usepackage[T1]{fontenc}
\usepackage{caption}
\usepackage{diagbox}
\mathchardef\hyphenmathcode=\mathcode`\-
\usepackage{amsbsy}
\usepackage{amsfonts}
\usepackage{amsmath}
\usepackage{amssymb}
\usepackage{amsthm}
\usepackage{cancel}
\usepackage[pdftex]{color}
\usepackage{enumerate}
\usepackage{float}
\usepackage{geometry, calc, color}%
\usepackage{indentfirst}
\usepackage{longtable}
\usepackage{multirow}
\usepackage{natbib}
\usepackage{graphicx}
\usepackage{url} 
\usepackage{xr-hyper}
\usepackage{hyperref}
\usepackage{booktabs}
\usepackage[linesnumbered,ruled,vlined]{algorithm2e}

\usepackage[dvipsnames]{xcolor}

\newtheorem{theorem}{Theorem}[section]
\newtheorem{definition}{Definition}[section]

\newtheorem{corollary}[theorem]{Corollary}
\newtheorem{remark}[theorem]{Remark}

\hypersetup{
  colorlinks=true,
  linkcolor=blue,
  citecolor=blue,
  filecolor=blue,
  urlcolor=blue,
}
\newcommand{\blind}{1}

\SetCommentSty{mycommfont}

\protect

\begin{document}

\def\spacingset#1{\renewcommand{\baselinestretch}%
{#1}\small\normalsize} \spacingset{1}

\spacingset{1.375} 

\if1\blind
{
  \title{\bf Time Series Analysis of Rankings: \\A GARCH-Type Approach}
  \author{Luiza S.C. Piancastelli$^\sharp$\footnote{\texttt{e-mail}: \href{mailto:luiza.piancastelli@ucdconnect.ie}{luiza.piancastelli@ucd.ie}}\,\, and Wagner Barreto-Souza$^\sharp$\footnote{\texttt{e-mail}: \href{mailto:wagner.barreto-souza@ucd.ie}{wagner.barreto-souza@ucd.ie}}
  \\
  {\normalsize \it $^\sharp$School of Mathematics and Statistics, University College Dublin, Belfield, Republic of Ireland}}
  \maketitle
} \fi
\if0\blind
{
  \bigskip
  \bigskip
  \bigskip
  \begin{center}
    {\LARGE\bf Time Series Analysis of Ranking Data}
\end{center}
  \medskip
} \fi
\date{}

\bigskip
\addtocontents{toc}{\protect\setcounter{tocdepth}{1}}

\begin{abstract}
Ranking data are frequently obtained nowadays but there are still scarce methods for treating these data when temporally observed.  The present paper contributes to this topic by proposing and developing novel models for handling time series of ranking data. We introduce a class of time-varying ranking models inspired by the Generalized AutoRegressive Conditional Heteroskedasticity (GARCH) models. More specifically, the temporal dynamics are defined by the conditional distribution of the current ranking given the past rankings, which are assumed to follow a Mallows distribution, which implicitly depends on a distance. Then, autoregressive and feedback components are incorporated into the model through the conditional expectation of the associated distances. Theoretical properties of our ranking GARCH models such as stationarity and ergodicity are established. The estimation of parameters is performed via maximum likelihood estimation when data is fully observed. We develop a Monte Carlo Expectation-Maximisation algorithm to deal with cases involving missing data. Monte Carlo simulation studies are presented to study the performance of the proposed estimators under both non-missing and missing data scenarios. A real data application about the weekly ranking of professional tennis players from 2015 to 2019 is presented under our proposed ranking GARCH models.

\end{abstract}


\section{Introduction}

Ranking data arises in situations where some order is attributed to a set of alternatives. For example, marketing researchers often ask judges to order options from best to worst to understand their preferences. Sports and competitions are another context where rankings are common. For instance, they can represent leaderboards, teams standing on long-term championships, or racing results. Ranking information is increasingly common in modern technology, playing important roles in search engines and recommendation systems. 

A complete ranking of a set of $k$ distinct alternatives with arbitrary labels $\{1, \ldots, k\}$ will be denoted by $\pi$. The latter is a point in the space of permutations of $k$ items, $\mathcal{P}_k$. The statistical analysis of ranking collections is a non-standard task where it must be acknowledged that data points carry multivariate relative information. A sample of $n$ observed rankings of $\{1, \ldots, k\}$ is represented as $\underline{\pi}$, where $\underline{\pi}$ is a $n \times k$ matrix. In addition to the unique challenges conveyed by $\pi \in \mathcal{P}_k$, it is common in practical situations that rankings are incomplete. An incomplete ranking observation is one that does not elicit the ordering of all elements of the item set and will be denoted by $\tilde{\pi}$. For instance, in top$-l$ elicitation, judges are asked to provide their first $l<k$ choices. This is a common approach for dealing with moderate to large item sets, frequently applied within educational (university rankings) and political (voting) settings. Another possibility is that $\tilde{\pi}$ is generated from pairwise choices. In this case, preferences among pairs $\{i, j\}, i \neq j, i,j \in \{1, \ldots, k\}$ are collected for different $i,j$ combinations. Most often, not all combinations necessary to produce a complete ranking are elicited. For example, consider $k=3$ and the set of decisions $\mathcal{C} = \{1 \succ 2, 1 \succ 3\}$ where $i \succ j$ indicates strict preference of $i$ over $j$. The observation of $\mathcal{C}$ generates an incomplete ranking as it is compatible with $\pi = (1,2,3)$ and $\pi = (1,3,2)$.

Statistical treatment of $\pi$ or $\tilde{\pi}$ aims to describe the rankings-generating mechanism via probabilistic models. This can take several directions that vary on their underlying assumptions. Some possibilities are the \textit{Thurstone, multistage} and \textit{distance-based models}. Under \textit{Thurstone} (or order-statistic) \textit{models}, it is assumed that there exist latent scores associated with items and a ranking is generated from their ordering. The Plackett-Luce \citep{plackett1975} model is the most famous within \textit{multistage models}, a class pinned by the independence of irrelative alternatives assumption (IIA). IIA states that the relative preference between two alternatives should not be affected by the introduction or removal of other alternatives. For instance, if a judge prefers option 1 over 2, they should continue to do so regardless of the presence of a third. Finally, \textit{distance-based models} assume the existence of a \textit{consensus ranking}, say $\pi_0$, and attribute probability mass according to the ranking's distances to $\pi_0$. The Mallows model \citep{mal1957} and its generalisations \citep{fliver1986} are prominent approaches in this class, highly regarded for their interpretability and parsimony. For a more detailed description of probability models for rankings, we recommend \cite{marden95}. 

Our contribution is situated within \textit{distance based} models, a class that has seen many recent developments. A significant focus in this area has been on \textit{rank aggregation}. For instance, the works by \cite{iru2018} and \cite{crispino2023} explore consensus estimation using the Mallows model with alternative distances. In \cite{meila2010} and \cite{crispino2017}, ranking aggregation from top-$l$ rankings and pairwise preferences with inconsistencies are addressed, respectively. Extensions of the Mallows model have also been employed to tackle quality and stability in rankings \citep{li2019} and to handle rank aggregation from non-homogeneous populations \citep{vitelli2018}. While the rank aggregation task is well-studied, the literature on dynamic rankings is still scarce.  
In practice, rankings that evolve over time arise naturally in numerous real-world contexts, reflecting the dynamic nature of preferences. For instance, in sports, studying the evolution of teams or players' rankings can offer insights into performance trends and competitive balance. In marketing and consumer trends, tracking changes in customer preferences through product rankings over time can guide marketing strategies and inventory planning. Other real-world examples of ranking time series are the dynamical ranking of web pages within search engines, the yearly rankings of academic institutions, and candidate rankings from sequential election polls. 

 To the best of our knowledge, there are only two approaches to handle time-varying rankings with stochastic dynamics. The first work in this area is due to \cite{asfetal2017}, where a time-varying ranking model based on the Mallows model is proposed under the Bayesian framework. Their model assumes conditional independence among the rankings and the temporal dynamics are driven by a further Mallows model for the consensus ranking (Markovian structure) and a Gaussian random walk for the scale parameter. A second approach was recently introduced by \cite{holzou22}, where the model is constructed based on the Plackett-Luce distribution with time-varying parameters following a Generalised AutoRegressive Score \citep{creetal2013} structure.


The chief goal of this paper is to introduce and explore a novel methodology for time series of rankings inspired by the Generalised AutoRegressive Conditional Heteroskedasticity (GARCH) models \citep{eng1982,bol1986}. Our modelling is built upon the Mallows model, which implicitly depends on distance. More specifically, we propose a reparameterisation of the Mallows distribution in terms of the mean parameter. Then, the temporal dynamic is defined by the conditional distribution of the current ranking given the past rankings, which is assumed to follow a reparameterised Mallows distribution, where autoregressive and feedback components are incorporated into the model through the conditional expectation of the associated distances. Our approach results in a GARCH-type dynamic for the sequence of distances. We call this novel class of time series by Ranking-GARCH (in short R-GARCH) models.
This formulation allows us to establish desirable stochastic properties for our model such as stationarity and ergodicity, and offers an intuitive dynamic in terms of the distances. In particular, the sample autocorrelation function (ACF) and partial ACF of the distances provide useful information on the lags to be considered in the modelling. Moreover, the ACF can be derived from the well-established results from ARMA processes since the distance process admits an ARMA representation. Another important contribution of this paper is the development of a Monte Carlo EM algorithm to deal with incomplete rankings.

The paper unfolds as follows. In Section \ref{sec:R-GARCH}, we introduce some needed notation, define our class of Ranking GARCH models, and derive some statistical properties. Section \ref{sec:mle} is devoted to model inference when the data is fully observed, while Section \ref{sec:incomplete} handles inference when missing data occurs. Simulated results are also presented to explore the finite-sample performance of the proposed estimators under both non-missing and missing data scenarios. We apply the proposed R-GARCH models to analyse the weekly ranking of professional tennis players from 2015 to 2019 in Section \ref{sec:app}. Concluding remarks are discussed in Section \ref{sec:conclusion}.

\section{Ranking GARCH Models}\label{sec:R-GARCH}

In this section, our proposed class of ranking time series models is defined and illustrated. As introduced previously, a complete ordering of an item set $\{1,\ldots,k\}$ is a point in $\mathcal{P}_k$, the space of permutations of $k$ labels. Observations in $\mathcal{P}_k$ are denoted by $\pi=(\pi(1),\ldots,\pi(k))$, where $\pi(i)$ is the rank given to item $i$, for $i=1,\ldots,k$. The inverse operation, $\pi^{-1}(i)$ is also well defined and returns the rank of item $i$. 

Construction of the R-GARCH approach will be based on the model by \cite{mal1957}, which is defined via the probability function
\begin{eqnarray}\label{mallows}
p(\pi|\theta)=\exp\{-\theta d(\pi,\pi_0)\}/\psi(\theta),\quad \pi\in\mathcal{P}_k,\quad \theta\in\mathbb R,
\end{eqnarray}
where $\pi_0$ is a fixed ranking known as the \textit{population consensus}, and $\theta$ is a concentration parameter. For $\pi,\pi_0\in\mathcal{P}_k$, $d(\cdot,\cdot)$ is a distance on $\mathcal{P}_k\times\mathcal{P}_k$, so $d(\pi,\pi_0)\geq0$ with equality holding if and only if $\pi=\pi_0$. For $\theta>0$, $\pi_0$ is the modal ranking. The limit cases $\theta\rightarrow0^+$ and $\theta\rightarrow\infty$ give us the uniform distribution (on $\mathcal{P}_k$) and the degenerate distribution at $\pi_0$, respectively. When $\theta<0$, $\pi_0$ is an antimode. We only consider the case $\theta>0$ in this paper, which also ensures model identifiability. Also, we focus on distances that are right-invariant \citep{cri1985} as often done in the literature. Mathematically, $d(\cdot,\cdot)$ is right-invariant if it satisfies $d(\pi,\pi_0)=d(\pi\circ\tau,\pi_0\circ\tau)$, $\forall\pi,\pi_0,\tau\in\mathcal P_k$, where $\pi\circ\tau$ denotes the composition of permutations $\pi$ and $\tau$ in the same space. In other words, under right-invariance, the distance is not affected by item-relabelling. A ranking $\pi$ following the Mallows model (\ref{mallows}) is denoted by $\pi\sim\mbox{Mallows}(\pi_0, \theta)$, which implicitly depends on the distance $d(\cdot,\cdot)$.

The choice of distance is a crucial point concerning Mallows model tractability. This is because the availability of an analytical form for the model's normalisation constant given by $\psi(\theta) = \sum_{\pi \in \mathcal{P}_k} \exp\{-\theta d(\pi,\pi_0)\}$ depends on $d(\cdot,\cdot)$. The majority of applications of the Mallows model focus on the Kendall and Hamming distances because closed-form solutions for $\psi(\theta)$ have been derived in the seminal work by \cite{fliver1986}. More recently in \cite{iru2018}, this has been achieved for the Cayley distance based on a result provided by \cite{fliver1986}. However, the derivation of these closed-form expressions is non-trivial. For instance, the Spearman and Footrule distances are often of interest but are choices that make (\ref{mallows}) an intractable likelihood (except for trivially small $k$). In this situation, the model inference is challenging and requires the use of $\psi(\theta)$ approximations or specialised methods. 

We now consider the Mallows model given in Equation (\ref{mallows}) reparameterised in terms of the mean $\mu=E(\pi)\equiv g(\theta)$, which will be the basis for our time series construction. We denote this reparameterised version as $\pi\sim\mbox{Mallows}(\mu,\pi_0)$. Let $\{\pi_t\}_{t\in\mathbb Z}$ be a sequence of random rankings on $\mathcal{P}_k$, with $\pi_t$ denoting a ranking of the item set at time $t$. To specify a GARCH-type probability model for $\pi_t$ at time $t$, we let $\pi_{t-1}$ to be the population consensus, and consider a reparametrisation of (\ref{mallows}) in terms of the distance mean $\mu_t$ conditional on the past rankings, which induces a dynamics for the spread parameter, now depending on $t$, say $\theta_t$. Evidently, $\mu_t$ is a function of the previous rank $\pi_{t-1}$ and the unknown spread parameter $\theta_t$. We denote this as $\mu_t = g(\pi_{t-1}, \theta_t)$, or simply $\mu_t = g(\theta_t)$ for right-invariant distances. Naturally, the form of $g(\cdot)$ is dictated by the type of distance, which we approach next. 
With these ingredients, a formal definition of our ranking time series models inspired by the GARCH processes is presented below. 

\begin{definition}[R-GARCH models]\label{def:R-GARCH}
Let $\{\pi_t\}_{t\in\mathbb Z}$ be a time series of rankings and $\mathcal F_{t-1}$ be the sigma-algebra generated by $\pi_{t-1},\pi_{t-2}\ldots$, for $t\in\mathbb Z$. We say that $\{\pi_t\}_{t\in\mathbb Z}$ is a Ranking-GARCH process if satisfies 
\begin{eqnarray}\label{temporal_dep}
\pi_t|\mathcal F_{t-1}&\sim&\mbox{Mallows}(\mu_t,\pi_{t-1}),\nonumber\\
\mu_t&\equiv&E\left(d(\pi_t,\pi_{t-1})|\mathcal F_{t-1}\right)=\phi_0+{\sum_{i=1}^p\phi_id(\pi_{t-i},\pi_{t-i-1})}+\sum_{j=1}^q\alpha_j\mu_{t-j},
\end{eqnarray}
for $t\in\mathbb Z$, with $p,q\in\mathbb N_0\equiv\{0,1,2,\ldots\}$, where $\phi_0>0$ is an intercept,  $\phi_1,\ldots,\phi_p\geq0$ are autoregressive parameters, and $\alpha_1,\ldots,\alpha_q\geq0$ are parameters related to effect of past conditional distance means. We denote $\{\pi_t\}_{t\in\mathbb Z}\sim\mbox{R-GARCH(p,q)}$.
\end{definition}\label{def:rgarch}

For a Ranking-GARCH process $\{\pi_t\}_{t\in\mathbb Z}$, $\mu_t$ is the conditional mean function of the distance $d(\pi_t,\pi_{t-1})$ given $\mathcal F_{t-1}$, which is time-varying and depends on the parameter vector $\boldsymbol\beta=(\phi_0, \boldsymbol{\phi}^\top, \boldsymbol{\alpha}^\top)^\top$, where $\boldsymbol{\phi} \equiv (\phi_1, \ldots, \phi_p)^\top$ and $\boldsymbol{\alpha} \equiv (\alpha_1, \ldots, \alpha_q)^\top$. The conditional probability of $\pi_t$ given $\mathcal{F}_{t-1}$, say $p(\pi_t|\mathcal{F}_{t-1})$ assumes the form (\ref{mallows}), which implicitly depends on $d(\cdot,\cdot)$, and with $\theta_t = g^{-1}(\mu_t)$ instead of $\theta$.

To explicitly compute the conditional probability function involved in Definition \ref{def:R-GARCH} in terms of $\mu_t$, it is necessary to compute $\theta_t=g^{-1}(\mu_t)$. Therefore, having a closed-form expression for the expected value of the distance plays is crucial for this task. For the R-GARCH model with Kendall distance, such a function $g(\cdot)$ is known and given in Eq. (\ref{eq:exp_kendall}) from Remark \ref{rem:kendall}. Although our simulation experiments will focus on this specification, the R-GARCH model is tractable for any $d(\cdot, \cdot)$ where $\psi(\theta)$ and $g(\theta)$ are available. In Remark \ref{rem:hamming}, we present explicit quantities under the Hamming distance (mean given in Eq. (\ref{eq:hamming})), which will also be considered in our data application.

\begin{remark}\label{rem:kendall}
For the numerical experiments in this paper, we will consider the Kendall distance. In this case, closed forms for moments of $d(\cdot, \cdot)$ and the normalisation constant $\psi(\theta)$ can be explicitly found in the literature. More specifically, for 
$\pi \sim \mbox{Mallows}(\pi_0,\theta)$ (classic parameterisation) with the Kendall distance, the expected value and variance of $d(\pi, \pi_0)$ are given respectively by
\begin{eqnarray}\label{eq:exp_kendall}
\mu\equiv g(\theta)\equiv E(d(\pi, \pi_0)) = \frac{k e^{-\theta}}{1-e^{-\theta}} - \sum_{j=1}^k \frac{j e^{-j\theta}}{1- e^{-j\theta}}
\end{eqnarray}
and 
$$\mbox{Var}(d(\pi, \pi_0)) = \frac{k e^{-\theta}}{ (1-e^{-\theta})^2} - \sum_{j=1}^k \frac{j^2 e^{-j\theta}}{(1- e^{-j\theta})^2},$$
for any modal ranking $\pi_0$; for instance, see \cite{fliver1986}. Also, the model's normalisation constant has the analytical form:
$$\psi(\theta) =\frac{\prod_{l=2}^{k}(1 - e^{-\theta l})}{(1-e^{-\theta})^{k-1}},$$
which does not depend on $\pi_0$ due to right-invariance of $d(\cdot, \cdot)$.
\end{remark}

\begin{remark}\label{rem:hamming}
For $\pi \sim \mbox{Mallows}(\pi_0, \theta)$ with the Hamming distance and by using the results on moment generating function provided by \cite{fliver1986}, we obtain that the expected value and variance of $d(\pi, \pi_0)$ are given respectively by
\begin{eqnarray}\label{eq:hamming}
\mu\equiv g(\theta)\equiv E(d(\pi, \pi_0)) = k-e^\theta\dfrac{A(k-1,\theta)}{A(k,\theta)}
\end{eqnarray}
and 
$$\mbox{Var}(d(\pi, \pi_0)) =e^\theta\dfrac{A(k-1,\theta)}{A(k,\theta)}+e^{2\theta}\left\{\dfrac{A(k-2,\theta)}{A(k,\theta)}+\left(\dfrac{A(k-1,\theta)}{A(k,\theta)}\right)^2\right\},$$
for any modal ranking $\pi_0$, where we have defined $A(k,\theta)\equiv\displaystyle\sum_{j=0}^k \frac{(e^\theta -1)^j}{j!}$. The analytical form of $\psi(\theta)$ in this case is
 $$\psi(\theta) = k! e^{-k\theta} A(k,\theta)$$
for any $\pi_0$ due to right-invariance of the Hamming $d(\cdot, \cdot)$, with $A(\cdot,\cdot)$ as defined above.

\end{remark}

With the given results in Remark \ref{rem:kendall} or Remark \ref{rem:hamming}, inverting the function $g(\cdot)$ requires finding the solution of $\mu_t = g(\theta_t)$ with respect to $\theta_t$. As this depends on the GARCH parameters, we adopt the notation $\theta_t(\phi_0, \boldsymbol{\phi}, \boldsymbol{\alpha})$. 
Finding $\theta_t(\phi_0, \boldsymbol{\phi}, \boldsymbol{\alpha})$ can be done as described in Algorithm \ref{alg:solve_theta} or with any other root-finding routines. Our pseudo-code provides a solution that defines a target function, which returns the squared difference between $g(\theta_t)$ and $\mu_t$. The latter is then minimised with respect to $\theta_t$ which can easily be accomplished with the \texttt{optim} routine in \texttt{R}. The algorithm is adapted to distances other than the Kendall by changing line 4 from Algorithm \ref{alg:solve_theta} to the appropriate $g(\cdot)$ form. 

\begin{algorithm}[ht!]
	\SetAlgoLined
	\SetKwInput{KwInput}{Input}                
	\SetKwInput{KwOutput}{Output}              
 	\SetKwFunction{FMain}{Main}
	\SetKwFunction{FSum}{Sum}
	\SetKwFunction{FTarget}{Target}
	
	\KwInput{ $\boldsymbol{\underline{\pi}}$, $\phi_0$, $\boldsymbol{\phi}$, $\boldsymbol{\alpha}$ }
	
	\For{$t \in (\max\{p,q\}+1):n$}{

	Compute $\mu_t = \phi_0 +\sum_{i=1}^p\phi_jd(\pi_{t-i},\pi_{t-i-1})+\sum_{j=1}^q\alpha_j\mu_{t-j}$

    	{\texttt{Target}{($\theta_t, \mu_t, k$}):}{ \textcolor{Cerulean}{// Define \texttt{Target} function}
		
		 $E = \frac{k e^{-\theta_t}}{1-e^{-\theta_t}} - \sum_{j=1}^k \frac{j e^{-j\theta_t}}{1- e^{-j\theta_t}}$;
 	
 	 \KwRet $(E-\mu_t)^2$
		
	}

	$\theta_t(\phi_0, \boldsymbol{\phi}, \boldsymbol{\alpha}) \leftarrow$ Minimize \FTarget{$\theta_t, \mu_t, k$} with respect to $\theta_t$ for the given $\mu_t$ and $k$;
 }

 	\KwOutput{${\theta}_{\max\{p,q\}+1}(\phi_0, \boldsymbol{\phi}, \boldsymbol{\alpha}), \ldots, \theta_n(\phi_0, \boldsymbol{\phi}, \boldsymbol{\alpha})$}

\caption{Algorithm to compute $\theta_t(\phi_0, \boldsymbol{\phi}, \boldsymbol{\alpha})$ for all $t$ under the Kendall distance.}\label{alg:solve_theta}
\end{algorithm}

We conclude this subsection by giving some remarks on the R-GARCH model for clarity of exposure. 

\begin{remark}\,\\
\vspace{-.8cm}
\begin{enumerate}[(i)]
    \item At time $t$, the ranking ($\pi_t$) has its temporal dependence driven by Equation (\ref{temporal_dep}), which states a GARCH-type structure for the sequence of distances $d_t\equiv d(\pi_t,\pi_{t-1})$, for $t\in\mathbb Z$. The conditional expectation of $d_t$ given $\mathcal F_{t-1}$ ($\mu_t$) depends on the previous distances $d_{t-1}\ldots,d_{t-p}$, which plays the rule of autoregressive (AR) components with respective AR coefficients $\phi_1,\ldots,\phi_p$. The feedbacks $\mu_{t-1},\ldots,\mu_{t-q}$ are also included to explain $\mu_t$ with respective coefficients $\alpha_1,\ldots,\alpha_q$. At this point, it is important to highlight that we are modelling the conditional mean of $d_t$ given $\mathcal F_{t-1}$ rather than the conditional variance considered in the well-established GARCH models. Note that the conditional variance is also modelled under our methodology since it depends on $\mu_t$; therefore, R-GARCH models have joint conditional mean/variance dynamics. Our approach resembles the integer-valued GARCH models (for instance, see \cite{feretal2006}), which are designed to handle count time series data and consider a time-varying conditional mean.
    
    \item For $p=q=0$, the spread parameter is a constant function of $\phi_0$. In this situation, there is no temporal dependence to be modelled and one should consider the classic Mallows model (assuming independence among the rankings) with mode $\pi_0$ and parameter $\theta = g^{-1}(\phi_0)$.  

    \item Although the R-GARCH model is presented in detail for Kendall and Hamming cases, its formulation for other distances $d(\cdot,\cdot)$ follows straightforwardly if the functions $g(\cdot)$ and $\psi(\cdot)$ (mean and normalisation constant) are available in closed form. Naturally, working with distances under which $\psi(\theta)$ and distribution moments are not analytical presents additional challenges. However this is not exclusive to the R-GARCH models, but common to any Mallows-type model involving them.   
\end{enumerate}    
\end{remark}

\subsection{Data Generation}

Simulated trajectories of the proposed model are explored in this section, which focuses on investigating the effects of the R-GARCH parameters. To this end, we can avail of data generation routines implemented for the Mallows model in the \texttt{R} package \texttt{BayesMallows} \citep{bayesmallows}. Our starting scenario considers the ranking of $k=10$ objects over $n=1000$ time points (sample size). We consider $p=1$ with $\phi_1 = 0.3$ and $q=0$. Different values of $\phi_0 = (3, 5, 7)$ are included initially so the conditional mean structure becomes $\mu_t = \phi_0 + \phi_1 d(\pi_{t-1}, \pi_{t-2})$. 

Three trajectories are simulated using the different $\phi_0$ values and are illustrated next. In Figure \ref{fig:sim1}, the conditional mean process over time is illustrated on the right, with different colors indicating the $\phi_0$ used for simulation. Naturally, higher values of $\phi_0$ produce higher conditional means, which indicates that larger deviations from $\pi_t$ to $\pi_{t-1}$ are supported. 
On the left-side panel of Figure \ref{fig:sim1}, histograms show the empirical distribution of the lag-one distances, i.e. $d(\pi_{t-1}, \pi_{t-2})$ for all $t$, obtained in each configuration. The histograms evidence how higher distance values are supported when increasing $\phi_0$, while small intercepts encourage small lag-one deviations.

\begin{figure}[ht!]
    \centering
    \includegraphics[width = 0.85\linewidth]{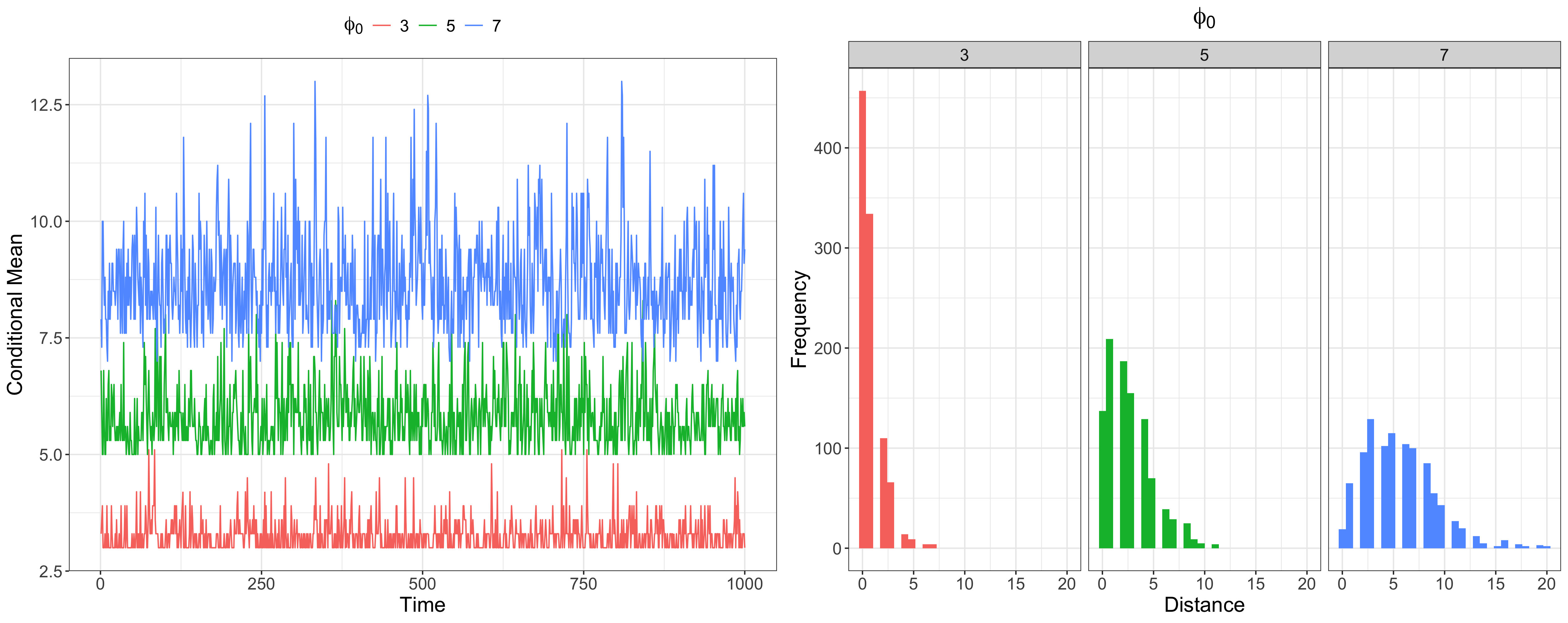}
    \caption{Conditional mean trajectories of three simulated ranking time series with $k=10$ items and length $n=1000$. Highest values of the intercept $\phi_0$ model larger values of $\mu_t=E\left(d(\pi_t,\pi_{t-1})|\mathcal F_{t-1}\right)$, for $t=2,\ldots,1000$. This is further illustrated with the histograms on the right, which show the frequency of $d(\pi_t,\pi_{t-1})$ observed under the trajectories obtained with $\phi_0 = 3,5,7$.}
    \label{fig:sim1}
\end{figure}

An alternative visualisation of the same simulated data is provided in Figure \ref{fig:orderings_in_time}. To produce this figure, rankings are converted to orderings so $\boldsymbol{\underline{\pi}}^{-1} \equiv (\pi^{-1}_1, \ldots, \pi^{-1}_k)$ is obtained. This allows us to explore the trajectory of rank assignments given to a fixed object over time. Figure \ref{fig:orderings_in_time} provides a visualisation of $\pi^{-1}_t(i)$ versus $t$ for all items $i=1, \ldots,10$, with $i$ indicated in each row. Repeating this exercise for the three simulated datasets results in the columns of Figure \ref{fig:orderings_in_time}, which vary in $\phi_0$ value. The alternative visualisation of the ranking time series provided in this figure highlights how smaller $\phi_0$ encourage less variation in the rankings received by each item over time. 

\begin{figure}[ht!]
    \centering
    \includegraphics[width = 0.95\linewidth]{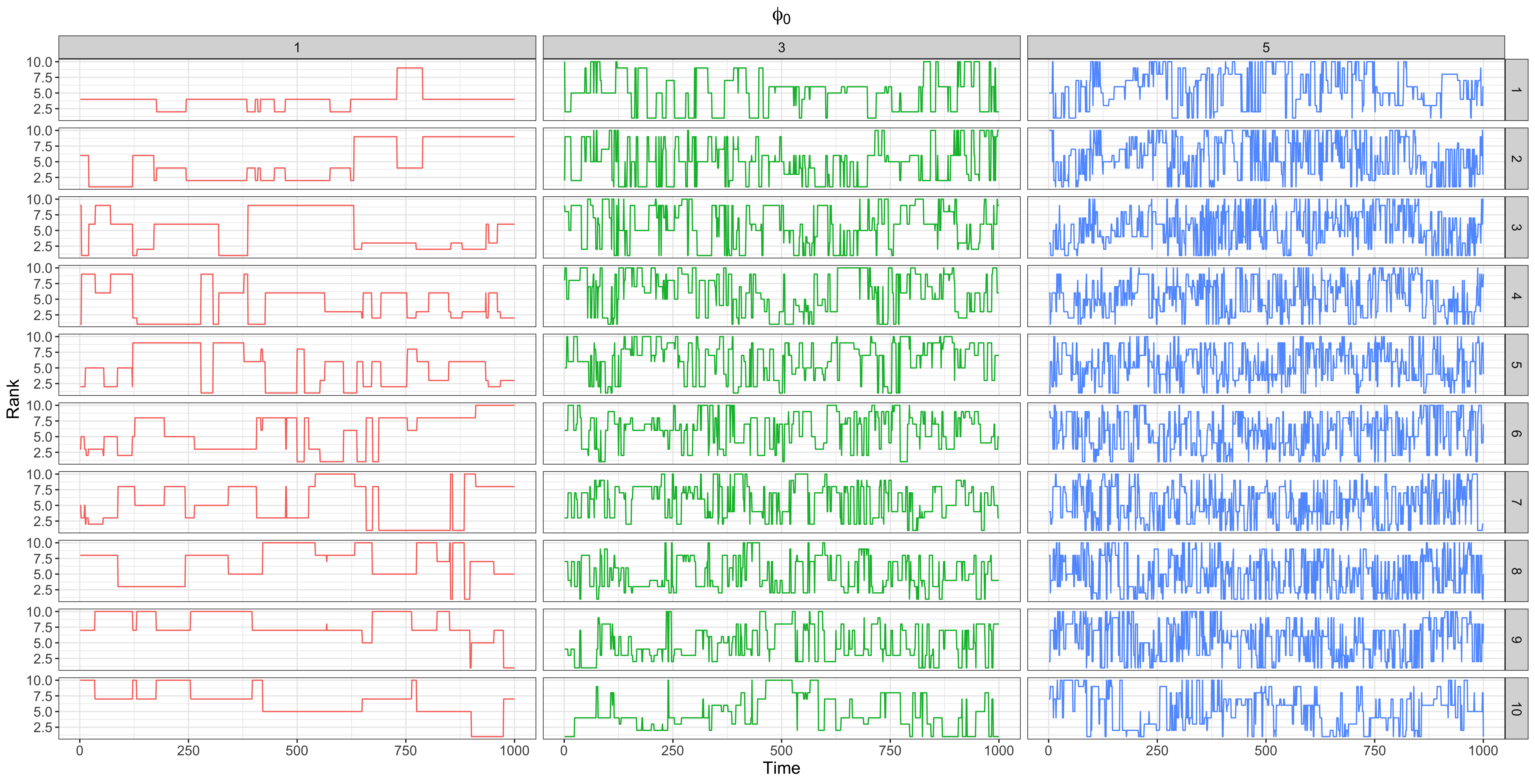}
    \caption{Time series of orderings given to items $i = 1, \ldots, 10$ (rows) in the trajectories simulated with $k=10, n=1000$, $\phi_1 = 0.3$ and $\phi_0 = 1, 3, 5$.}\label{fig:orderings_in_time}
\end{figure}

Now that the role of $\phi_0$ in R-GARCH models is well understood, further settings aim at exploring the parameters that control the time dynamics. We continue studying the ten-item setting and compare R-GARCH models that employ $p=1, q=0$ and $p=1, q=1$. The conditional conditional expectation structure is now $\mu_t = \phi_0 + \phi_1 d(\pi_{t-1}, \pi_{t-2}) + \alpha_1 \mu_{t-1}$ and parameter values are $\phi_0 = 3, \phi_1 = 0.3$ and $\alpha_1 =0$ or $\alpha_1 = 0.3$.  In this study, one thousand trajectories are generated from each R-GARCH(1,0) ($p=1, q=0$) and R-GARCH(1,1) ($p=1, q=1$) configurations.

The empirical autocorrelation of lags one to ten of the conditional mean process $\{\mu_t \}_{t=2}^n$ is computed and stored for each trajectory. Figure \ref{fig:acf_models} illustrates the results with boxplots, with lags indicated in the horizontal axis and configurations represented with different colors. With the addition of the moving average term, the ACF is increased when compared to the case $q=0$ as expected. By looking at Figure \ref{fig:acf_models}, the behaviour of the ACFs resembles that of ARMA models with exponential decay. In the next subsection, we show that the sequence of distances admits an ARMA representation, which theoretically justifies this statement.


\begin{figure}[ht!]
    \centering
    \includegraphics[width = 0.8\linewidth]{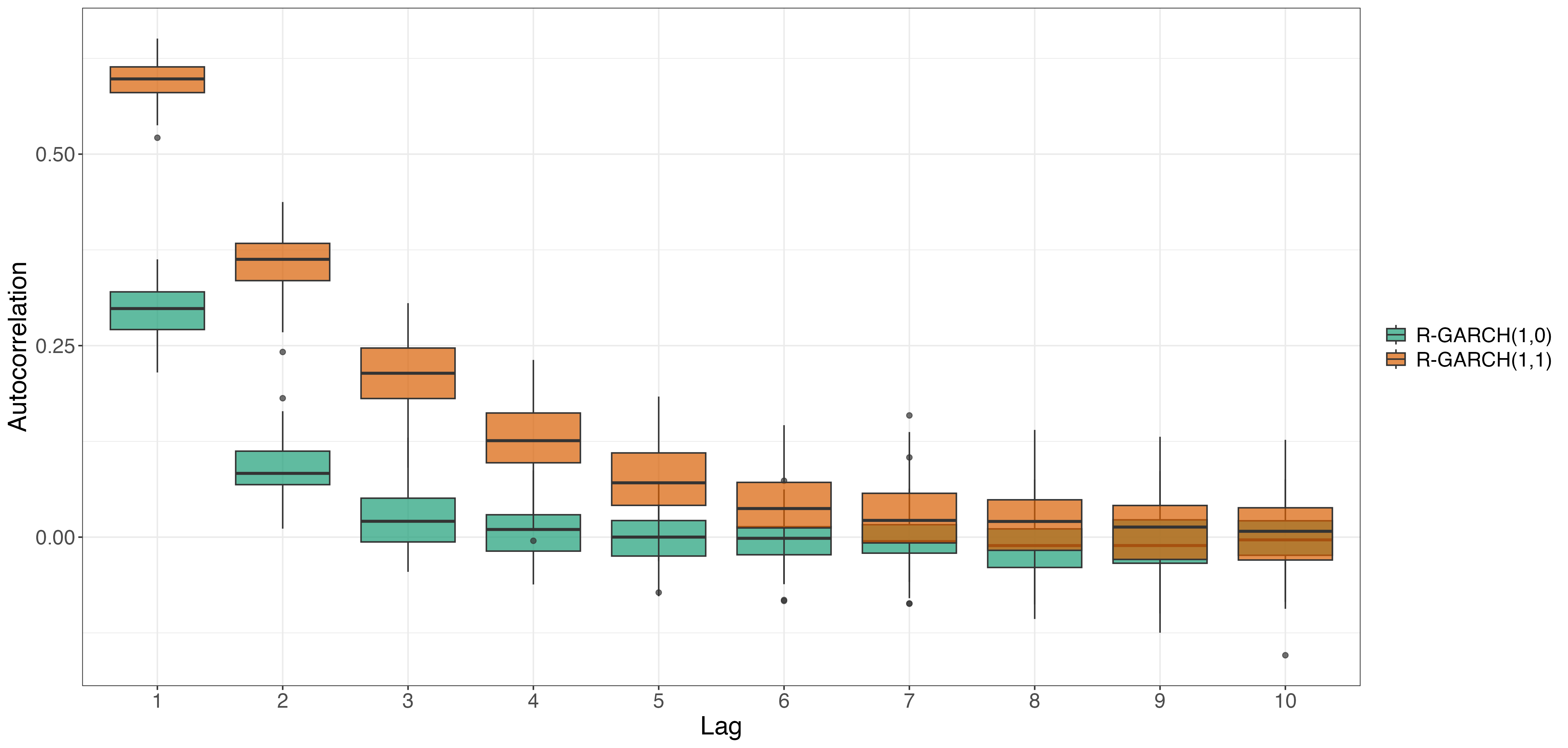}
    \caption{Boxplots of empirical autocorrelation of lags one to ten from one hundred R-GARCH(1,0) and R-GARCH(1,1) simulated trajectories.} \label{fig:acf_models}
\end{figure}

\subsection{Stationarity and Ergodicity}

In this subsection, we establish stationarity and ergodicity properties for the sequence of distances $d_t\equiv d(\pi_t,\pi_{t-1})$, for $t\in\mathbb Z$.

\begin{theorem}\label{thm:meanstat}
The sequence of distances $\{d_t\}_{t\in\mathbb Z}$ is first-order stationary if $\displaystyle\sum_{i=1}^p\phi_i+\displaystyle\sum_{j=1}^q\alpha_j<1$.
\end{theorem}
\begin{proof}
By computing the expectation in both sides of Equation (\ref{temporal_dep}) and by using the notation $d_t\equiv d(\pi_t,\pi_{t-1})$, we have that 
\begin{eqnarray}\label{proof:aux}
E(\mu_t)=E(d_t)=\phi_0+\sum_{i=1}^p\phi_iE(d_{t-i})+\sum_{j=1}^q\alpha_jE(d_{t-j}),
\end{eqnarray}
where we have used the fact that $E(\mu_{t-j})=E[E(d_{t-j}|\mathcal F_{t-j-1})]=E(d_{t-j})$ for $j=1,\ldots,q$. Denote $\lambda_t\equiv E(d_t)$ for $t\in\mathbb Z$. Then, the above equation can be expressed as 
$$\lambda_t=\phi_0+\sum_{i=1}^{\max(p,q)}(\phi_i+\alpha_i)\lambda_{t-i},$$
where $\phi_i\equiv0$ for $i>p$ and $\alpha_j\equiv0$ for $j>q$. In terms of the backshift operator $B$, it follows that
$$\left(1-\sum_{i=1}^{\max(p,q)}(\phi_i+\alpha_i)B^i\right)\lambda_t=\phi_0,$$
which is a homogeneous difference equation. It has a stable and finite solution that does not depend on $t$ if $1-\displaystyle\sum_{i=1}^{\max(p,q)}(\phi_i+\alpha_i)z^i\neq0$ $\,\,\forall |z|\leq1$. Since $\phi_1,\ldots,\phi_p$ and $\alpha_1,\ldots,\alpha_q$ are non-negative, this condition is equivalent to $\displaystyle\sum_{i=1}^{\max(p,q)}(\phi_i+\alpha_i)<1$.
\end{proof}

\begin{corollary}\label{cor:mean}
Under the first-order stationarity of $\{d_t\}_{t\in\mathbb Z}$, we have that 
$$E(d_t)=E(d(\pi_t,\pi_{t-1}))=\dfrac{\phi_0}{1-\displaystyle\sum_{i=1}^p\phi_i-\displaystyle\sum_{j=1}^q\alpha_j}.$$
\end{corollary}
\begin{proof}
The result is immediately obtained by using Equation (\ref{proof:aux}) and the fact that $E(d_t)=E(d_{t-s})$ for all $s\in\mathbb Z$ under first-order stationarity.
\end{proof}

The next result states that the first-order stationary condition is also enough for the second-order stationary of the sequence of distances.

\begin{theorem}\label{thm:covstat}
A necessary and sufficient condition for $\{d_t\}_{t\in\mathbb Z}$ to be second-order stationary is that $\displaystyle\sum_{i=1}^p\phi_i+\displaystyle\sum_{j=1}^q\alpha_j<1$.
\end{theorem}

\begin{proof}
To show the sufficiency, we follow the strategy by \cite{rydshe1999}, where they showed that the temporal dynamic of a Poisson time series model (INGARCH) satisfies an ARMA representation.

Define $\epsilon_t\equiv d_t-\mu_t$, for $t\in\mathbb Z$. The sequence of distances $d_t$ satisfies the following equations:
\begin{eqnarray}\label{proof:arma}
d_t&=&\mu_t+(d_t-\mu_t)=\mu_t+\epsilon_t\nonumber\\
   &=&\phi_0+\sum_{i=1}^p\phi_id_{t-i}+\sum_{j=1}^q\alpha_j\mu_{t-j}+\epsilon_t\nonumber\\
   &=&\phi_0+\sum_{i=1}^p\phi_id_{t-i}+\sum_{j=1}^q\alpha_j(d_{t-j}-\epsilon_{t-j})+\epsilon_t\nonumber\\
   &=&\phi_0+\sum_{i=1}^{\max(p,q)}(\phi_i+\alpha_i)d_{t-i}-\sum_{j=1}^q\alpha_j\epsilon_{t-j}+\epsilon_t,
\end{eqnarray}
for $t\in\mathbb Z$, where $\phi_i\equiv0$ for $i>p$ and $\alpha_j\equiv0$ for $j>q$ as before. We will now argue that $\{\epsilon_t\}_{t\in\mathbb Z}$ is a martingale difference sequence (MDS). First, from Theorem \ref{thm:meanstat} and Corollary \ref{cor:mean}, we have that $\epsilon_t$ is integrable under the condition $\displaystyle\sum_{i=1}^p\phi_i+\displaystyle\sum_{j=1}^q\alpha_j<1$. More specifically, $E(|\epsilon_t|)=E(|d_t-\mu_t|)\leq E(d_t)+E(\mu_t)=2E(d_t)=\dfrac{2\phi_0}{1-\displaystyle\sum_{i=1}^p\phi_i-\displaystyle\sum_{j=1}^q\alpha_j}<\infty$. Further, $E(\epsilon_t|\mathcal F_{t-1})=E(d_t|\mathcal F_{t-1})-E(\mu_t|\mathcal F_{t-1})=\mu_t-\mu_t=0$, almost surely. Therefore, $\{\epsilon_t\}_{t\in\mathbb Z}$ is indeed an MDS with respect to $\{\mathcal F_{t-1}\}_{t\in\mathbb Z}$. This fact and Equation (\ref{proof:arma}) give us that $\{d_t\}_{t\in\mathbb Z}\sim\mbox{ARMA}(\max(p,q),q)$. Hence, it follows from the well-established ARMA results in the literature that $\{d_t\}_{t\in\mathbb Z}$ is second-order stationary if all roots of the AR polynomial $\Phi(z)=1-\displaystyle\sum_{i=1}^{\max(p,q)}(\phi_i+\alpha_i)z^i$ lie outside the unit circle, which along with the fact that the coefficients are all non-negative, it is equivalent to $\displaystyle\sum_{i=1}^{\max(p,q)}(\phi_i+\alpha_i)<1$. Finally, the necessity comes from the first-order stationarity, which is valid under the assumption of the theorem. Under first-order stationarity, the mean of $d_t$ given in Corollary \ref{cor:mean} is necessarily positive since $d_t>0$. This immediately gives us that $\displaystyle\sum_{i=1}^p\phi_i+\displaystyle\sum_{j=1}^q\alpha_j<1$ is necessary.
\end{proof}

\begin{remark}
Due to representation (\ref{proof:arma}), the variance and autocovariance functions of $\{d_t\}_{t\in\mathbb Z}$ can be directly obtained from well-established results from ARMA processes provided in time series textbooks (e.g. \cite{brodav1991}). Hence, in particular, for an R-GARCH(1,1) model, we have that
\begin{eqnarray*}
\gamma(0)\equiv\mbox{Var}(d_t)=\sigma_\epsilon^2\dfrac{1-2\phi_1\alpha_1-\alpha_1^2}{1-(\phi_1+\alpha_1)^2}
\end{eqnarray*}
and
\begin{eqnarray*}
\rho(h)\equiv\mbox{corr}(d_{t+h},d_t)=\dfrac{\phi_1[1-\alpha_1(\phi_1+\alpha_1)]}{1-2\phi_1\alpha_1-\alpha_1^2}(\phi_1+\alpha_1)^{h-1},\quad\mbox{for}\quad h\geq1.
\end{eqnarray*}
where $\sigma_\epsilon^2\equiv\mbox{Var}(\epsilon_t)$.
\end{remark}

We end this subsection by establishing the ergodicity of the sequence of distances for an R-GARCH(1,1) model.

\begin{theorem}\label{thm:erg}
    Let $\{d_t\}_{t\in\mathbb Z}$ be the sequence of distances associated with an R-GARCH(1,1) model. Then, for $\phi_1+\alpha_1<1$, $\{d_t\}_{t\in\mathbb Z}$ is geometric ergodic.
\end{theorem}

\begin{proof}
We will show that Tweedie's criterion \citep{meytwe1993} is satisfied for the test function $V(x)=1+x$, for $x>0$. This result and the Markovianity of $d_t$ (under the R-GARCH(1,1) model framework) give us the desirable geometric ergodicity.

It follows that
\begin{eqnarray*}
E(V(\mu_t)|\mu_{t-1}=\mu)&=&1+E(\mu_t|\mu_{t-1}=\mu)\\
&=&1+E(\phi_0+\phi_1 d +\alpha_1\mu)\\
&=&\dfrac{1+\phi_0+\phi_1 E(d) +\alpha_1\mu}{1+\mu}V(\mu).
\end{eqnarray*}
where $d$ is the distance associated with the Mallows model. By definition (model parameterisation), we have that $E(d)=\mu$. Hence, $\dfrac{1+\phi_0+\phi_1 E(d) +\alpha_1\mu}{1+\mu}=\dfrac{1+\phi_0+\mu(\phi_1  +\alpha_1)}{1+\mu}\rightarrow\phi_1+\alpha_1<1$ as $\mu\rightarrow\infty$, where we have used the hypothesis of the theorem that $\phi_1+\alpha_1<1$. Further, note that $\mu$ is bounded away from zero since the smallest reachable point of $\mu_t$ is $\mu_*=\dfrac{\alpha_0}{1-\alpha_1}$ (fix point of the skeleton $\mu_t=\alpha_0+\alpha_1\mu_{t-1}$). Therefore, there exist constants $\kappa_1\in(0,1)$, $\kappa_2>0$ and $\bar\mu>\mu_*$ such that
\begin{eqnarray}\label{tweediecriterion}
E(V(\mu_t)|\mu_{t-1}=\mu)\leq (1-\kappa_1)V(\mu)+\kappa_2 I\{\mu\in(\mu_*,\bar\mu)\},
\end{eqnarray}
where $I\{\cdot\}$ is the indicator function, and the existence of $\kappa_2$ follows from the fact that the function $1+\phi_0+\mu(\phi_1 +\alpha_1)$ is bounded on the interval $(\mu_*,\bar\mu)$. Since Tweedie's criterion (\ref{tweediecriterion}) is satisfied, the geometric ergodicity of $\{d_t\}_{t\in\mathbb Z}$ holds.
\end{proof}


The following two sections are concerned with parameter estimation, where the main aim is to learn $\phi_0, \boldsymbol{\phi}, \boldsymbol{\alpha}$ from observed ranking time series. Section \ref{sec:mle} is devoted to model inference when $\underline{\boldsymbol\pi}$ are complete ranks, while the missing value case is covered in Section \ref{sec:incomplete}.


\section{Inference Under Complete Rankings}\label{sec:mle}

Let $\pi_1,\ldots,\pi_n$ be an observed time series of rankings from an $\mbox{R-GARCH}(p,q)$ model as given in Definition \ref{def:R-GARCH}. Let $\boldsymbol\beta=(\phi_0,\boldsymbol\phi^\top,\boldsymbol\alpha^\top)^\top$ be the parameter vector. The log-likelihood function $\ell(\boldsymbol\beta)$ assumes the form
\begin{eqnarray}\label{eq:loglik}
\ell(\boldsymbol\beta)=\sum_{t=m+1}^n\ell_t(\boldsymbol\beta)\equiv\sum_{t=m+1}^n\{-\theta_td_t-\log\psi(\theta_t)\}=-\sum_{t=m+1}^n\{g^{-1}(\mu_t)d_t+\log\psi(g^{-1}(\mu_t))\},
\end{eqnarray}
where $m=\max(p,q)$, $d_t\equiv d(\pi_t,\pi_{t-1})$, and $\mu_t$ as in Equation (\ref{temporal_dep}). 


Fitting the R-GARCH model to observed complete ranking can be straightforwardly done via conditional maximum likelihood estimation. This consists of optimising the log-likelihood function with respect to $\boldsymbol\beta$, that is, the conditional maximum likelihood estimator (CMLE) of $\boldsymbol\beta$ is given by $\widehat{\boldsymbol\beta}=\mbox{argmax}_{\boldsymbol\beta} \ell(\boldsymbol\beta)$. Standard maximisation routines such as \texttt{R}'s \texttt{optim} can be employed on this task, producing parameter estimates at a cheap computational cost.  

The score function associated with $\ell(\boldsymbol\beta)$, say ${\bf U}(\boldsymbol\beta)=\partial\ell(\boldsymbol\beta)/\partial\boldsymbol\beta$, has components given by
\begin{eqnarray*}
\dfrac{\partial\ell(\boldsymbol\beta)}{\partial\beta_j}=\sum_{t=m+1}^n\dfrac{\partial\ell_t(\boldsymbol\beta)}{\partial\beta_j}=-\sum_{t=m+1}^n\left\{d_t+\dfrac{\psi'(\theta_t)}{\psi(\theta_t)}\right\}\dfrac{\partial\theta_t}{\partial\mu_t}\dfrac{\partial\mu_t}{\partial\beta_j},
\end{eqnarray*}
for $j=1,\ldots,p+q+1$, where $\theta_t=g^{-1}(\mu_t)$ and $\psi'(\theta)=d\psi(\theta)/d\theta$. By using that $\dfrac{\partial\theta_t}{\partial\mu_t}=\dfrac{1}{g'\left(g^{-1}(\mu_t)\right)}=\dfrac{1}{g'(\theta_t)}$ and $\dfrac{d\log\psi(\theta_t)}{d\theta_t}=E(d_t|\mathcal F_{t-1})=\mu_t$, we obtain that the score function can be expressed by
\begin{eqnarray}\label{score}
\dfrac{\partial\ell(\boldsymbol\beta)}{\partial\beta_j}=-\sum_{t=m+1}^n\dfrac{d_t-\mu_t}{g'(\theta_t)}\dfrac{\partial\mu_t}{\partial\beta_j},\quad j=1,\ldots,p+q+1.
\end{eqnarray}

The score function (\ref{score}) admits a simple matrix representation: ${\bf U}(\boldsymbol\beta)=-{\bf X}\boldsymbol D({\bf d}-{\boldsymbol\mu})$, where ${\bf X}$ is a matrix $p+q+1\times n-m$ with elements $\dfrac{\partial\mu_t}{\partial\beta_j}$, ${\bf D}=\mbox{diag}\{1/g'(\theta_{m+1}),\ldots,1/g'(\theta_n)\}$, ${\bf d}=(d_{m+1},\ldots,d_n)^\top$ and $\boldsymbol\mu=(\mu_{m+1},\ldots,\mu_n)^\top$. Furthermore, this analytical form can be provided in the optimisation to improve the performance of \texttt{optim} in \texttt{R}.

We now discuss obtaining the standard errors of the maximum likelihood estimates. First, we argue that our R-GARCH models satisfy the second Bartlett identity $E\left(\dfrac{\partial\ell_t(\boldsymbol\beta)}{\partial\boldsymbol\beta}\dfrac{\partial\ell_t(\boldsymbol\beta)}{\partial\boldsymbol\beta^\top}\bigg|\mathcal F_{t-1}\right)=-E\left(\dfrac{\partial^2\ell_t(\boldsymbol\beta)}{\partial\boldsymbol\beta\partial\boldsymbol\beta^\top}\bigg|\mathcal F_{t-1}\right)$.
For $k,l=1,\ldots,p+q+1$, we have that 
\begin{eqnarray}\label{eq:information}
-E\left(\dfrac{\partial^2\ell_t(\boldsymbol\beta)}{\partial\beta_k\partial\beta_l}\bigg|\mathcal F_{t-1}\right)=\dfrac{1}{\left(g'(\theta_t)\right)^2}\dfrac{d^2\log\psi(\theta_t)}{d\theta_t^2}\dfrac{\partial\mu_t}{\partial\beta_k}\dfrac{\partial\mu_t}{\partial\beta_l}=\dfrac{1}{g'(\theta_t)}\dfrac{\partial\mu_t}{\partial\beta_k}\dfrac{\partial\mu_t}{\partial\beta_l},
\end{eqnarray}
where we have used the fact $\mu_t=E(d_t|\mathcal F_{t-1})$ and $\dfrac{d^2\log\psi(\theta_t)}{d\theta_t^2}=g'(\theta_t)$, the last one following from exponential family properties. By using that $\mbox{Var}(d_t|\mathcal F_{t-1})=g'(\theta_t)$, we obtain that 
\begin{eqnarray*}
E\left(\dfrac{\partial\ell_t(\boldsymbol\beta)}{\partial\beta_k}\dfrac{\partial\ell_t(\boldsymbol\beta)}{\partial\beta_l}\bigg|\mathcal F_{t-1}\right)=\dfrac{1}{\left(g'(\theta_t)\right)^2}\mbox{Var}(d_t|\mathcal F_{t-1})\dfrac{\partial\mu_t}{\partial\beta_k}\dfrac{\partial\mu_t}{\partial\beta_l}=\dfrac{1}{g'(\theta_t)}\dfrac{\partial\mu_t}{\partial\beta_k}\dfrac{\partial\mu_t}{\partial\beta_l},
\end{eqnarray*}
for $k,l=1,\ldots,p+q+1$, which is equal to (\ref{eq:information}), and therefore it confirms our claim about the second Bartlett identity. The total information matrix, say ${\bf K}_n(\boldsymbol\beta)$, can be expressed as
\begin{eqnarray*}
{\bf K}_n(\boldsymbol\beta)=\sum_{t=m+1}^nE\left(-\dfrac{\partial^2\ell_t(\boldsymbol\beta)}{\partial\boldsymbol\beta\partial\boldsymbol\beta^\top}\bigg|\mathcal F_{t-1}\right)={\bf X}{\bf D}{\bf X}^\top,
\end{eqnarray*}
where the matrices ${\bf X}$ and ${\bf D}$ are defined below Equation (\ref{score}).

Under stationarity and ergodicity of $\{d_t\}_{t\in\mathbb Z}$, the asymptotic normality of the MLEs can be established: $\sqrt{n}({\widehat{\boldsymbol\beta}}-\boldsymbol\beta)\stackrel{d}{\longrightarrow} N({\boldsymbol0},\boldsymbol\Sigma)$ as $n\rightarrow\infty$, where $\boldsymbol\Sigma=\mbox{plim}_{n\rightarrow\infty}n^{-1}{\bf K}_n(\boldsymbol\beta)$, with $\mbox{plim}$ denoting the limit in probability. Note that we formally found the condition to ensure stationarity and ergodicity of the distance sequence under an R-GARCH(1,1) model in Theorem \ref{thm:erg}. The standard errors of the maximum likelihood estimates can be obtained by the diagonal elements of the inverse of the total information matrix, that is ${\bf K}_n^{-1}(\boldsymbol\beta)=({\bf X}{\bf D}{\bf X}^\top)^{-1}$.

To evaluate the finite sample performance of the R-GARCH maximum likelihood estimators (MLEs) $\widehat{\boldsymbol{\beta}} \equiv (\widehat{\phi}_0, \widehat{\boldsymbol{\phi}}^\top, \widehat{\boldsymbol{\alpha}}^\top)^\top$, some simulation experiments are carried out. To this end, R-GARCH trajectories are generated under different parameter configurations and sample sizes. Our first experiment takes $k=10, \phi_0 = 1, \phi_1 = 0.4$ and $n=200,500$. With each $n$, one thousand trajectories are simulated, and $\boldsymbol{\widehat{\beta}}$ is computed. Results are summarised in Figure \ref{fig_mles_sim1} and Table \ref{tab_mles_sim1}. On the right, boxplots illustrate the distribution of $\boldsymbol{\widehat{\beta}} = (\widehat{\phi}_0, \widehat{\phi}_1)^\top$ by sample size, with vertical dashed lines indicating the true parameter values used to generate the data. As expected, estimates are centered at the true values and variability decreases with the increase in sample size. A numerical summary of the same experiment is displayed in Table \ref{tab_mles_sim1} where the Monte Carlo estimate, standard deviation, and mean squared error (MSE) are enclosed. 

\begin{figure}[!htb]
    \centering
    \begin{minipage}{0.55\textwidth}
        \centering
        \includegraphics[width = 1\linewidth]{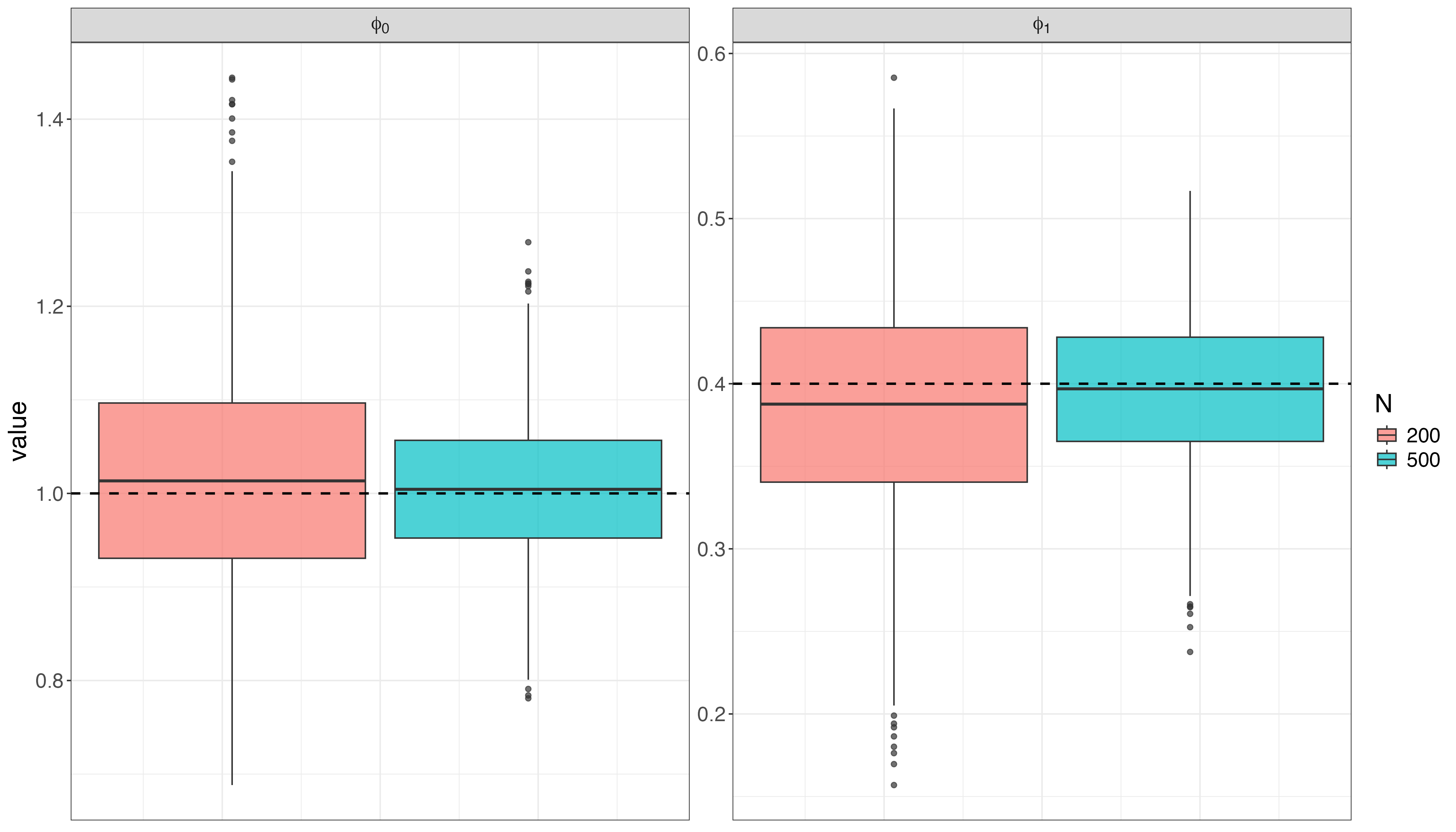}
             \captionof{figure}{Boxplots of maximum likelihood estimates fitted to one thousand R-GARCH data sets simulated with $k=10, \phi_0 = 1, \phi_1 = 0.4$. This is repeated considering the sample sizes $n=200$ and $n=500$ which are shown in different colors. True parameter values are indicated by vertical dashed lines.}\label{fig_mles_sim1}
    \end{minipage}\hfill
     \begin{minipage}{.4\textwidth}
        \centering
        \small
        \renewcommand{\arraystretch}{1}
\begin{tabular}{@{}lllll@{}}
\toprule
                          & \textbf{$n$} & \textbf{Mean} & \textbf{SD} & \textbf{MSE} \\ \midrule
\multirow{2}{*}{$\phi_0 =1$} & 200        & 1.014         & 0.125       & 0.016        \\
                          & 500        & 1.005         & 0.079       & 0.006        \\
\multirow{2}{*}{$\phi_1 = 0.4$}   & 200        & 0.387         & 0.070       & 0.005        \\
                          & 500        & 0.396         & 0.047       & 0.002        \\ \midrule 
\end{tabular}
     \captionof{table}{Monte Carlo mean, standard deviation, and Mean Squared Error (MSE) associated with the first simulation experiment.}\label{tab_mles_sim1}
    \end{minipage}
\end{figure}

A second simulation setting is performed with a larger item set $k=20$, including AR and feedback effects in the model. Parameter values are set to $\phi_0= 3, \phi_1 = 0.2, \alpha_1 = 0.3$ and the same sample sizes $n= 200, 500$ and Monte Carlo replications are explored. Results are illustrated with boxplots in Figure \ref{fig_mles_sim2}, and numerical summaries can be found in Table \ref{tab_mles_sim2}.

 As before, all model parameters are well estimated and there is a reduction in the standard errors and MSE as $n$ increases.  

\begin{figure}[!htb]
    \centering
    \begin{minipage}{0.55\textwidth}
        \centering
        \includegraphics[width = 1\linewidth]{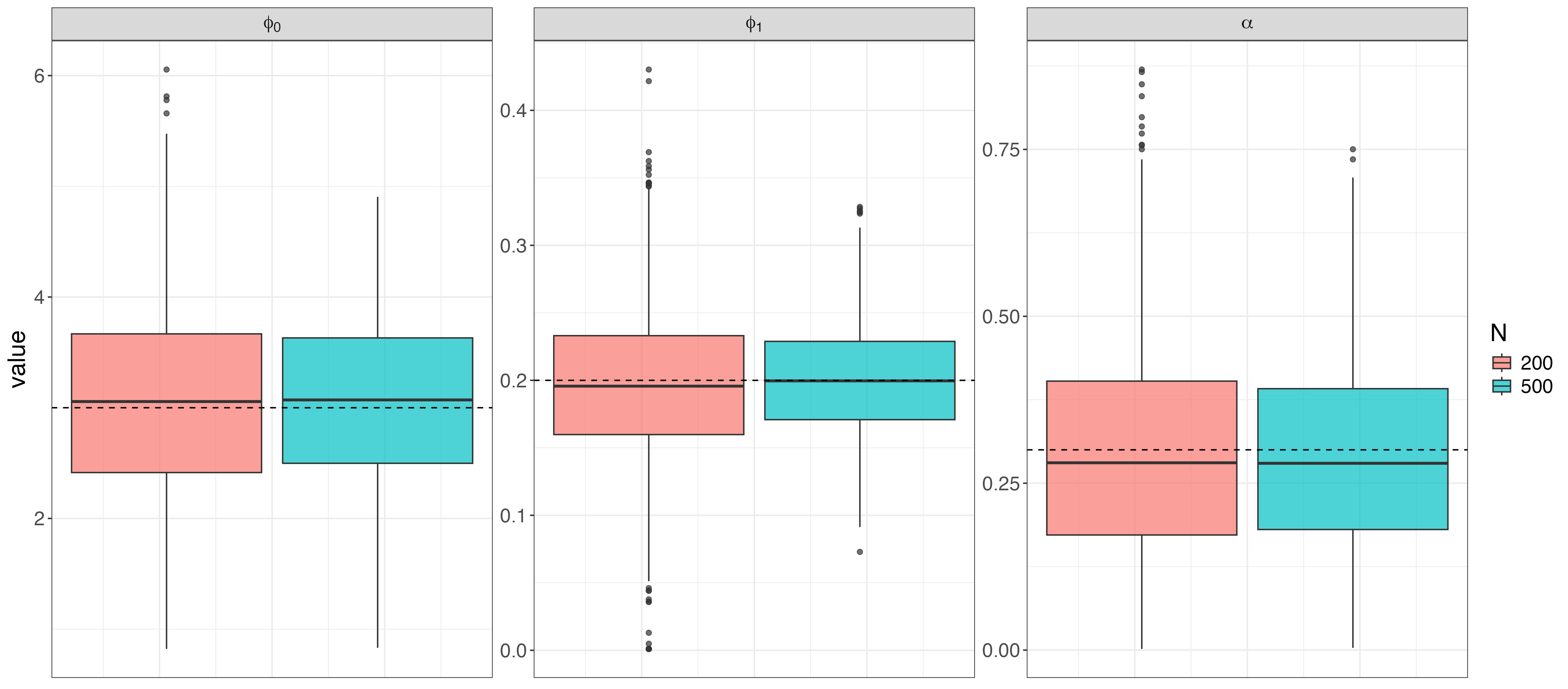}
        \captionof{figure}{Boxplots of maximum likelihood estimates fitted to one thousand R-GARCH data sets simulated with $k=20, \phi_0 = 3, \phi_1 = 0.2$ and $\alpha_1 = 0.3$. This is repeated considering the sample sizes $n=200$ and $n=500$ which are shown in different colors. Vertical dashed lines indicate true parameter values.}\label{fig_mles_sim2}
    \end{minipage}\hfill
     \begin{minipage}{.4\textwidth}
        \centering
        \small
        \renewcommand{\arraystretch}{1}
\begin{tabular}{@{}lllll@{}}
\toprule
                          & \textbf{$n$} & \textbf{Mean} & \textbf{SD} & \textbf{MSE} \\ \midrule
\multirow{2}{*}{$\phi_0=3$} & 200        & 3.063         & 0.875       &  0.768       \\
                          & 500        & 3.067         & 0.782       &   0.615      \\
\multirow{2}{*}{$\phi_1 = 0.2$}   & 200        & 0.196         & 0.058       &   0.003      \\
                          & 500        & 0.200         & 0.043       &     0.002    \\ 

\multirow{2}{*}{$\alpha_1 = 0.3$}   & 200        & 0.294         & 0.160       &    0.026     \\
                          & 500        & 0.288         & 0.143       &     0.020    \\ \midrule 
\end{tabular}
          \captionof{table}{Monte Carlo mean, standard deviation, and Mean Squared Error (MSE) associated with the second simulation experiment.}\label{tab_mles_sim2}
              \end{minipage}
\end{figure}

The studies presented in this section provide evidence of adequate performance of the conditional maximum likelihood estimation for obtaining point estimates of the R-GARCH parameters under the scenarios considered. The results validate the proposed inferential procedure for complete rankings, showing that the CMLEs perform adequately in practice.

\section{Estimation Under Incomplete Rankings}\label{sec:incomplete}

In many situations, the observed data $\pi_1, \ldots, \pi_n$ may involve incompleteness in a way that some (or all) $\pi_i$ do not rank the entire item set. For example, in top$-l$ elicitation, judges choose their top $l$ favorite alternatives, so positions $\pi(l+1), \ldots, \pi(k)$ are not observed. It is also common to evaluate choices by presenting a judge with pairs $\{i, j\}, i \neq j$, where $i, j \in \{1, \ldots, k\}$. In this setting, the pairwise preference is $i \succ j$ or $j \succ i$ but most often, one does not make all pairwise choices necessary to induce a full ranking. For instance, let $k=3$ and $\mathcal{S}$ denote the set of observed pairwise comparisons. If $\mathcal{S} = \{ 1 \succ 2$, $1 \succ 3$\} is collected, there are two complete rankings that are compatible with it, these being $\pi = (1,2,3)$ and $\pi = (1,3,2)$. In what follows, the notation $\tilde{\pi}$ shall indicate that $\pi$ is incomplete and $R(\tilde{\pi})$ be the set of compatible complete ranks ($R(\tilde{\pi}) = \{(1,2,3), (1,3,2)\}$ in the example above).

The most common approach to handle incompleteness in Mallows-type models is to assume that there exists a latent complete rank associated with $\tilde{\pi}$. The latter is assumed to be a random variable following some prior distribution that is consistent with the observed piece $\tilde{\pi}$. In the simplest case, ${\pi}$ follows an Uniform distribution over $R(\tilde{\pi})$ and $p({\pi}_j) = 1/n_{mis}(\tilde{\pi}_j)!$, where $n_{mis}(\tilde{\pi}_j)$ is the number of missing entries in $\tilde{\pi}_j$. 

In the R-GARCH context, having incomplete ranks becomes a somewhat more complex problem in comparison to when observations in $\underline{\pi}$ are independent. This is because our model formulation involves the sequence of distances, which cannot be computed if missingness is present. The following subsection is devoted to the challenging problem of performing inference in this scenario. We will refer to the Expectation-Maximisation (EM) algorithm, outlining in detail all ingredients necessary to handle incomplete R-GARCH observations.


\subsection{Monte Carlo EM Algorithm}

The Expectation-Maximisation (EM) algorithm is perhaps the most well-known tool for performing statistical inference in the presence of missing information. Consider an arbitrary observed ranking time series $\underline\pi^{o}$ that contains only partial information, and a joint probability model $p_o(\cdot|\boldsymbol\beta)$ associated with it. Our goal is to estimate the parameter vector $\boldsymbol\beta$, but $p_o(\cdot|\boldsymbol\beta)$ can only be evaluated under complete observations. The unobserved rankings (missing observations) are here denoted by $\widetilde\pi^{u}$, which will play the rule of auxiliary random variables in our EM algorithm. The joint probability function of $\widetilde\pi^{u}$ will be denoted by $p_u(\cdot)$. The distribution of the latent variable(s) ${\pi}_t$ will be assumed to be discrete Uniform over $\mathcal{R}(\widetilde{\pi}^u_t)$, as explained above. Moreover, independence among the missingness will be assumed for simplicity. The complete data is then denoted by $\underline{\pi}=(\underline\pi^{o},\widetilde\pi^{u})$.

For the R-GARCH model, the complete log-likelihood function, say $\ell^c(\boldsymbol\beta)$, assumes the form (\ref{eq:loglik}).
An important remark to make is that computation of $\ell_t(\boldsymbol\beta)$ in (\ref{eq:loglik}) depends on completeness of $\pi_t,\ldots,\pi_{\max(p,q)-1}$. For example, an R-GARCH(2,0) model requires that $\pi_{t-1}, \pi_{t-2}$ and $\pi_{t-3}$ are also complete in addition to $\pi_t$. From Definition \ref{def:rgarch}, it is obvious that if either $\pi_t$ or $\pi_{t-1}$ have missing entries, $\ell_t(\boldsymbol\beta)$ cannot be obtained since $d({\pi}_t, {\pi}_{t-1})$ is unavailable. In addition, computation of $\theta_t$ involves the lagged distances $d({\pi}_{t-i}, {\pi}_{t-i-1})$, so completeness of observations $\pi_{t-2}$ and $\pi_{t-3}$ is also required when $(p,q)=(2,0)$.

We here propose a Monte Carlo EM-type (MCEM) algorithm due to the inability of the computation of both observed log-likelihood function and conditional expectations involved in the EM-algorithm. To handle these situations, the MCEM was first introduced in \cite{wei1990} in the context of genetics. In this version, the $Q$-function to be maximised is approximated using simulations from the conditional distribution. MCEM variants differ in how the simulations are done, with some important implementations being the Stochastic Approximation EM (SAEM), Markov Chain Monte Carlo EM (MCMCEM), and Sequential Monte Carlo EM (SMCEM).

Let us start by describing the Monte Carlo E-step of our algorithm. If $t$-th term of the complete log-likelihood function involves a partial ranking, we complete the ranking by generating samples from the uniform distribution as specified above conditional on the partial ranking and by using $\boldsymbol\beta^{(r-1)}$ as the parameter vector in the generation; the superscript here refers to the $(r-1)$-th step of the MCEM algorithm. We consider $M$ imputations and then replace the $t$-th term of the complete log-likelihood function with the average of the imputed terms, which we denote by $\widetilde\ell_t(\boldsymbol\beta)$. No action is required if the $t$-th term of the complete log-likelihood function does not involve missingness; in this case, $\widetilde\ell_t(\boldsymbol\beta)=\ell_t(\boldsymbol\beta)$. Hence, we define the $Q$-function to be maximised as 
\begin{eqnarray}\label{eq:Q}
Q(\boldsymbol\beta;\boldsymbol\beta^{(r-1)})=\sum_{t=m+1}^n\widetilde\ell_t(\boldsymbol\beta).
\end{eqnarray}

It is important to note that $\widetilde\ell_t(\boldsymbol\beta)$ implicitly depends on $\boldsymbol\beta^{(r-1)}$, which denotes the MCEM estimate of $\boldsymbol\beta$ in the $r$-th loop of the algorithm. The M-step consists of maximising (\ref{eq:Q}) with respect to $\boldsymbol\beta$; the maximiser is denoted by $\boldsymbol\beta^{(r)}$.  At iteration $r$ of the MCEM algorithm, convergence is assessed based on some pre-specified criterion that compares $\boldsymbol\beta^{(r)}$ and $\boldsymbol\beta^{(r-1)}$. 
 
The strategy outlined in our R-GARCH MCEM algorithm is based on the implementations proposed by \cite{casella2001} and \cite{chan1995}. These works deal carefully with the crucial aspects of any MCEM implementation which are the \textit{number of auxiliary samples} $M$ and the \textit{stopping criterion}. Extra care must be taken when setting MCEM stopping rules, in comparison to the standard EM. This is because the algorithm must now account for the fact that there is a Monte Carlo error involved in $Q(\cdot, \cdot)$, which is determined by $M$. To approach this, \cite{casella2001} and \cite{chan1995} suggest to monitor the difference in the log-likelihood along iterations of the algorithm. The idea is that if the MLE has been reached, the change $\Delta \ell(\boldsymbol{\beta}^{(r)}, \boldsymbol{\beta}^{(r-1)}) \equiv \ell(\boldsymbol{\beta}^{(r)}) - \ell(\boldsymbol{\beta}^{(r-1)})$ will vary around zero. Since we are unable to compute the marginal likelihood under missingness, this difference is instead estimated using $\ell^c(\cdot)$. To this end, the auxiliary draws obtained in the E-step are retained and $\Delta \ell(\boldsymbol{\beta}^{(r)}, \boldsymbol{\beta}^{(r-1)})$ is estimated as
\begin{equation}\label{eq:ll_change}
    \widehat{\Delta} \ell^c(\boldsymbol{\beta}^{(r)}, \boldsymbol{\beta}^{(r-1)}) = \frac{1}{M}\sum_{j=1}^M \log p({\underline{\pi}}_j|\boldsymbol{\beta}^{(r)})-
    \frac{1}{M}\sum_{j=1}^M \log p({\underline{\pi}}_j|\boldsymbol{\beta}^{(r-1)}),
\end{equation}
where $\underline{\pi}_1, \ldots, \underline{\pi}_M$ denote $M$ replications of imputed ranking trajectories and $p(\cdot|\boldsymbol{\beta})$ is its associate joint probability function, which depends on $\boldsymbol{\beta}$.

Then, a possible stopping criterion around (\ref{eq:ll_change}) is to assess if a confidence interval of the form $(\eta -L \sigma, \eta + L \sigma)$ contains zero, where $\eta$ and $\sigma^2$ are the mean and variance of $\widehat\Delta \ell^c(\boldsymbol{\gamma}^{(r)}, \boldsymbol{\gamma}^{(r-1)})$, and $L$ is a non-negative constant to be specified. These unknown quantities are obtained by estimating $\eta$ with Equation (\ref{eq:ll_change}), and $\widehat{\sigma}^2$ with the empirical variance of $\dfrac{p({\underline{\pi}}_1|\boldsymbol{\beta}^{(r-1)})}{p({\underline{\pi}}_1|\boldsymbol{\beta}^{(r)})}, \ldots, \dfrac{p({\underline{\pi}}_M|\boldsymbol{\beta}^{(r-1)})}{p({\underline{\pi}}_M|\boldsymbol{\beta}^{(r)})}$.

\cite{chan1995} suggest $L$ be a number such as 4, in such a way that Chebychev's inequality guarantees that with a probability of at least $100(1- 1/L^2)\%$, the true difference lies within the interval. For further details, please refer to Section 2.3 from \cite{chan1995}. Another option is to create the interval based on the Central Limit Theorem and adopt a standard normal quantile in place of $L$.  

Finally, convergence is accepted if the interval contains zero and $\sigma < \epsilon$, guaranteeing that $\widehat{\Delta} \ell^c(\boldsymbol{\beta}^{(r)}, \boldsymbol{\beta}^{(r-1)})$ is estimated with minimum precision $\epsilon>0$. The latter effectively determines the sample size $M$, and the algorithm stops when $\widehat{\Delta} \ell^c(\boldsymbol{\beta}^{(r)}, \boldsymbol{\beta}^{(r-1)})$ is within $2L\sigma$ of zero. To satisfy this condition and optimise efficiency, we adopt the strategy of starting with moderate $M$ and increasing this value along with the algorithm's iterations.

This concludes our algorithm specification, for which pseudo-code is given in Algorithm \ref{alg:mcem}, which iterates between a Monte Carlo E-Step, maximisation of the estimated $Q$-function, and a convergence assessment based on a prespecified precision $\epsilon>0$ and the estimated confidence interval for $\widehat\Delta \ell^c(\boldsymbol{\beta}^{(r)}, \boldsymbol{\beta}^{(r-1)})$.

\begin{algorithm}
\DontPrintSemicolon
\SetAlgoLined
\SetNoFillComment
\LinesNotNumbered 
\KwIn{Observed incomplete data $\tilde{{\underline{\pi}}}$, $M$, $\epsilon$, $L$ and initial guess $\boldsymbol{\beta}^{(0)}$;}

\vspace{0.3cm}

Let $\mathcal{S}({\underline{{\pi}}})$ be the set of rows where $(\pi_t, \ldots,\pi_{t-\max(p,q)-1})$ involves missingness;

$c \leftarrow \mbox{FALSE}$; $r \leftarrow 0$;

\While{$c = \mbox{FALSE}$}{

\tcp*[l]{1. Monte Carlo Expectation Step}

\For{$j \in \mathcal{S}({\underline{\pi}})$}{
Sample $\boldsymbol{\pi}_j \equiv \widetilde{\pi}^1_j, \ldots, \widetilde{\pi}^M_j$ from $p(\pi_j) \equiv \mbox{Uniform}(\mathcal{R}(\tilde{\pi}_j))$;

Store $\boldsymbol{{\pi}}_j$;

}

\textcolor{gray}{// ${\underline{{\pi}}}_1 \equiv (\boldsymbol{\pi}_1^1, \ldots, \boldsymbol{\pi}_N^1), \ldots, {\underline{{\pi}}}_M \equiv (\boldsymbol{\pi}_1^M, \ldots, \boldsymbol{\pi}_N^M)$ complete data sets are obtained;}

\tcp*[l]{2. Maximisation Step}

Maximise $Q(\boldsymbol{\beta}; \boldsymbol{\beta}^{(r-1)})$ given in Eq. (\ref{eq:Q}) with respect to $\boldsymbol{\beta}$;

Set $\boldsymbol{\beta}^{(r)}$ as the solution of this maximization problem;
 
  \vspace{0.2cm}
  
\tcp*[l]{3. Convergence check}

Estimate $\widehat{\eta} = \widehat{\Delta}\ell^c(\boldsymbol{\beta}^{(r)}, \boldsymbol{\beta}^{(r-1)})$ and $\widehat{\sigma}^2$ using ${\underline{{\pi}}}_1, \ldots, {\underline{{\pi}}}_M$;

Compute the interval $\mathcal{I} = (\widehat{\eta} - L \widehat{\sigma}, \widehat{\eta} + L \widehat{\sigma})$;

\uIf{$\sigma \leq \epsilon$ \mbox{and} $0 \in \mathcal{I}$ }{
    $c = \mbox{TRUE}$ \texttt{\textcolor{cyan}{ convergence indicator}} \;
} \ElseIf{$\hat{\sigma}>\epsilon$}{
 $M = M + \lceil M/3 \rceil$ \;
 $r = r +1;$
} \Else{
$r = r +1;$
}

}

\KwOut{$\boldsymbol{\beta}^{(r)}, M, r$; }

\caption{Monte Carlo EM algorithm to fit the R-GARCH model parameters under incomplete ranking trajectories.}\label{alg:mcem}
\end{algorithm}

\subsection{MCEM Simulation Study}

A small simulation experiment is carried out to investigate the performance of the proposed inferential procedure for incomplete rankings. For this purpose, ranking time series data will be simulated from the R-GARCH model with different percentages of missingness. We set the data configuration to $k= 10, n = 500, \phi_0 = 3$ and $\phi_1 = 0.5$. Either 10\% or 30\% of missing entries are then introduced uniformly at random in the simulated data set as follows. First, a row index $t^*$ is drawn from $\{1, \ldots, n\}$. The number of rankings to be deleted from $\pi_{t^*}$ is then simulated between $\{2, \ldots, k/2\}$ and the positions to be removed are chosen at random. The index $t^*$ is excluded from the set of available rows and the procedure continues until the missingness percentage is approximately the target. The number of Monte Carlo replications is 100 per missingness percentage, $M$ is initialised to 200, and $\epsilon$ is set to 0.05.

The empirical density functions of the Monte Carlo parameter estimates are shown in Figure \ref{fig:mcem_experiment}, where different colors indicate 10\% (in red) or 30\% (in blue) missing entries and dashed lines indicate the parameter value used to generate the (complete) data. As expected, uncertainty increases slightly as more entries are removed. This can be seen clearly on the left-hand side, where the density associated with the 10\% case is sharper at the true value of $\phi_0$. In terms of the autoregressive parameter, there is a small downward bias which increases with the missingness percentage. We conjecture that this is due to the independence assumption made on the distribution of the latent complete ranks. Introducing time dependence in the latter would likely help mitigate this, but this would imply a more complicated and computationally expensive E-step. In our view, results are satisfactorily close to the true $\phi_1$, so we choose to avoid this extra hurdle. Anyway, we believe that this point deserves further investigation in a future paper.

\begin{figure}
    \centering
    \includegraphics[width = 0.8\linewidth]{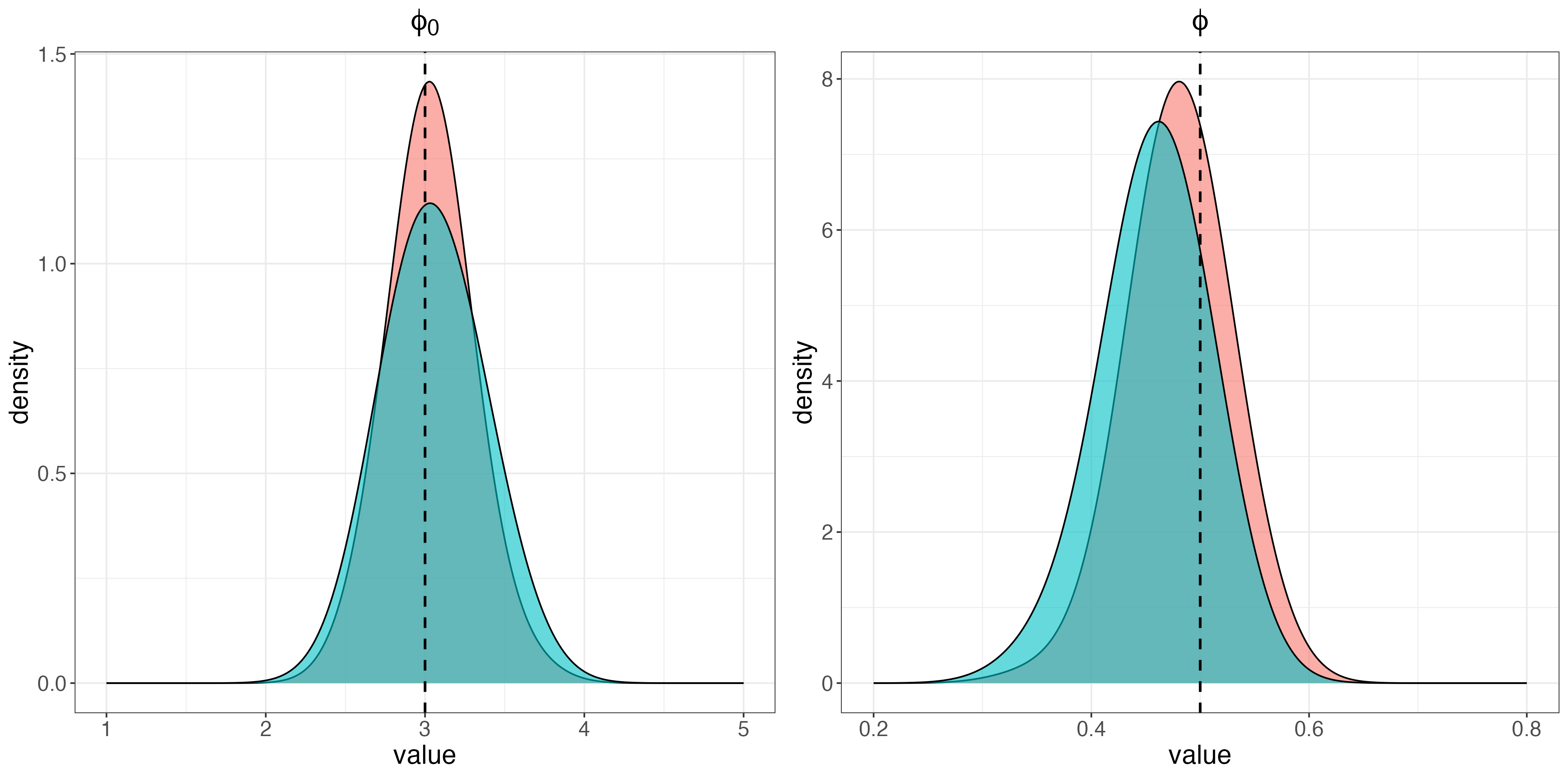}
    \caption{Empirical density functions of the R-GARCH(1,0) parameter estimates using simulated ranking time series data with missingness via MCEM (Algorithm \ref{alg:mcem}). One hundred data sets are simulated with $k= 10, n = 500, \phi_0 = 3$, and $\phi_1 = 0.5$ and either 10\% or 30\% of random missing entries. The different colors indicate 10\% (in red) or 30\% (in blue) of missingness, and the true parameter values and marked with vertical dashed lines.} \label{fig:mcem_experiment}
\end{figure}

\section{Application to Tennis Rankings}\label{sec:app}

Data containing the weekly rankings of professional male singles tennis players was sourced online from \hyperlink{https://www.kaggle.com/datasets/mimoopoo/atp-tennis-rankings-1990-to-2019}{Kaggle}, originally scrapped from the Association of Tennis Professionals (ATP) \hyperlink{https://www.atptour.com/en}{website}. It contains weekly rankings of the top one hundred athletes from 1990 to 2019. Naturally, there are changes in players entering or exiting the rankings across this period, with the total number of distinct individuals that appear at least once in the data being 718. To achieve some homogeneity of the player set, we decide to filter the data to the 2015-2019 period. This is done so that the effect of players disappearing from the ranking due to retirement, or emerging from having recently joined the sport competitively is reduced. This subset involves $n=172$ rankings and 202 players.

Our first aim is to analyse the evolution of players who were ranked in all 172 weeks, a subset of $k=31$ individuals. To this end, the relative ranking of such athletes is computed and our observed data becomes ${\underline{\pi}} \equiv (\pi_1, \ldots, \pi_{172})$, where $\pi_i \in \mathcal{P}_{31}$. Inferential procedures proposed for complete ranking time series in Section \ref{sec:mle} are applied to ${\underline{\pi}}$ under different values of $p$ and $q$. The R-GARCH model is fitted to different combinations of $p, q$, with the Hamming and Kendall distances. AIC and BIC are computed to perform order selection within the Kendall and Hamming R-GARCH models. We highlight that it is not sensible to perform distance selection under model information criteria because the R-GARCH can be written in terms of the distance trajectory and this changes according to the metric. The empirical autocorrelation function (ACF) and partial autocorrelation function (PACF) of the latter are displayed in Figure \ref{fig:acf_pacf_tennis}, which gives us an idea about the order to be considered. For both the Kendall and Hamming R-GARCH models, we fit combinations of $p, q \in \{0, \ldots, 3\}$, and compute AIC and BIC. The order $(p,q)=(3,0)$ is selected for both Kendall and Hamming R-GARCH models, based on either AIC or BIC. 

\begin{figure}
    \centering
    \includegraphics[width = 0.85\linewidth]{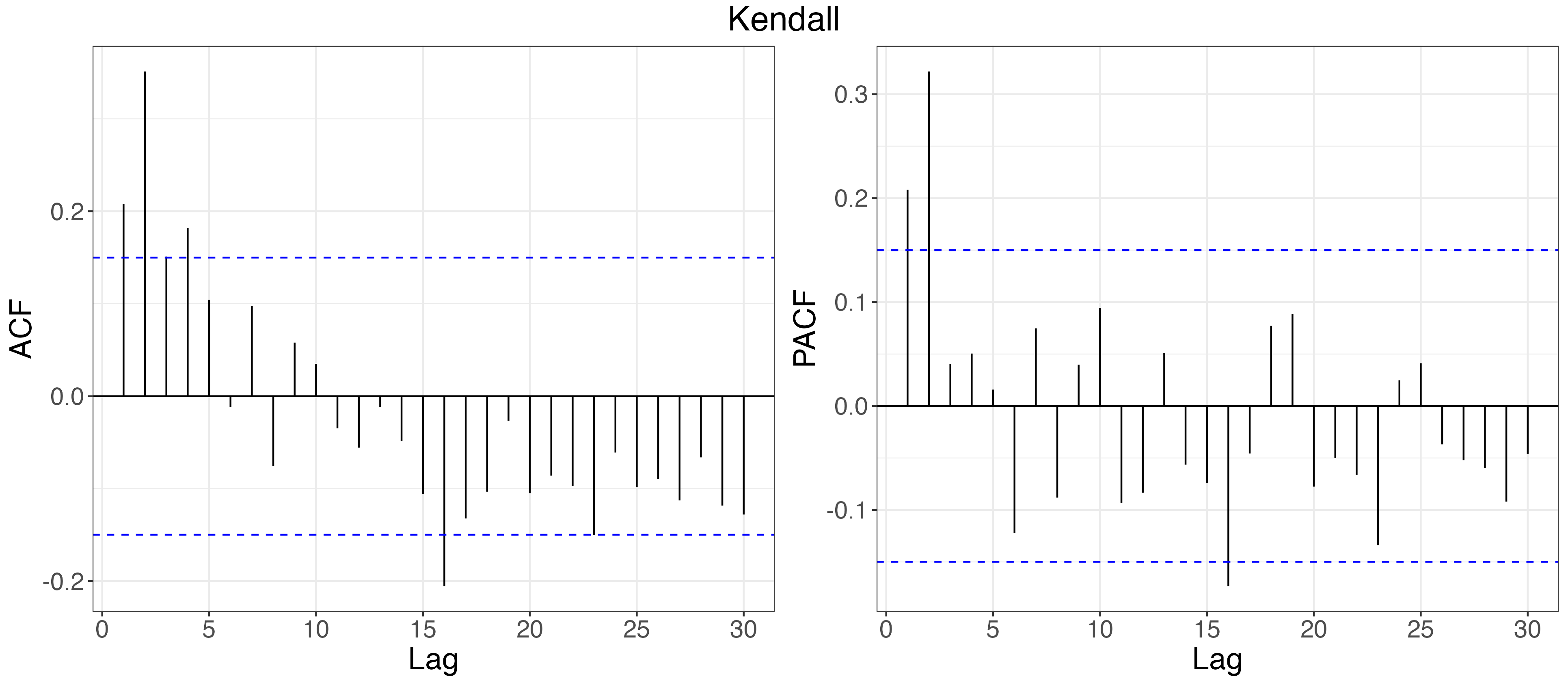}
        \includegraphics[width = 0.85\linewidth]{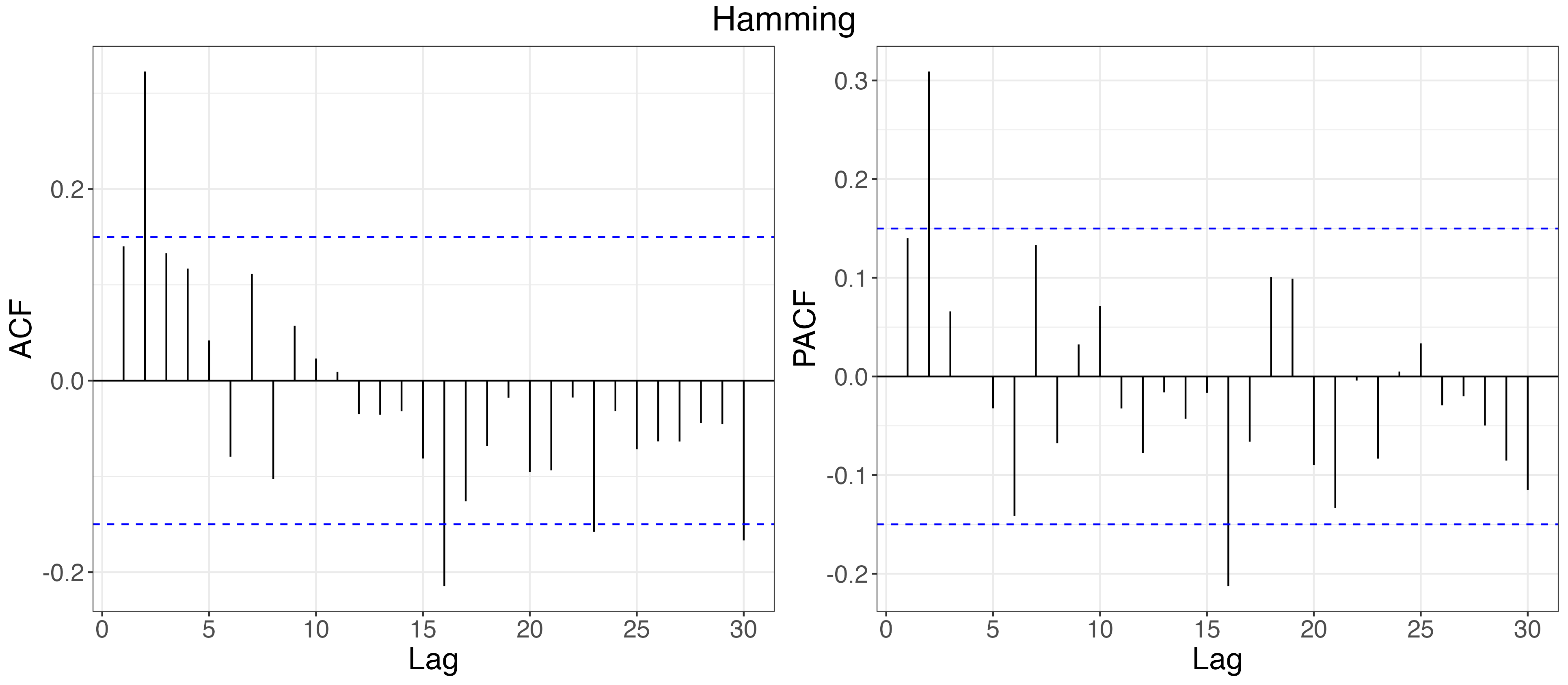}
    \caption{Autocorrelation function (ACF) and partial autocorrelation function (PACF) of the lag-one Kendall and Hamming distance trajectories $d(\pi_t, \pi_{t-1}), t = 2, \ldots, 172$, computed for the tennis rankings data.} \label{fig:acf_pacf_tennis}
\end{figure}

The models indicated by AIC and BIC are further inspected in Table \ref{tab:fitted_complete_tennis}, which contains the MLEs and their standard errors. 
An analysis of residuals is employed to compare the Kendall and Hamming R-GARCH(3,0) models. To this end, their fitted mean trajectories $\{ \hat{\mu}_t \}_{t=2}^{172}$ are obtained and these are used to compute the ordinary residuals $\hat{r}_t = d(\pi_t, \pi_{t-1}) - \hat{\mu}_t)$, for $t=5,\dots,172$.  Figure \ref{fig:residuals} displays the boxplot and ACF of $ \hat{r}_t$ computed from the models under comparison. To produce the boxplots on the left, the Hamming residuals are scaled to match the range of the Kendall ones. This is done by multiplying $\hat{r}_t$ by the factor ${k \choose 2}/k$. The latter is the ratio of distances' maximum values, which are ${k \choose 2}$ and $k$ respectively for the Kendall and Hamming forms. Once this is done, we compare the values of ordinary residuals from the two models via boxplots. These plots indicate that smaller deviations are obtained with the Kendall distance. The ACF of residuals shows that both models have captured well the source of autocorrelation. Overall, the residual analysis supports the decision to favour the Kendall R-GARCH(3,0) model in the analysis of tennis player's weekly rankings.

\begin{table}[]
\centering
\begin{tabular}{@{}lllll@{}}
\toprule
\multicolumn{1}{c}{} & \multicolumn{1}{c}{$\phi_0$} & \multicolumn{1}{c}{$\phi_1$} & \multicolumn{1}{c}{$\phi_2$} & \multicolumn{1}{c}{$\phi_3$} \\ \midrule
Kendall              & 32.080 (2.345)               & 0.274 (0.021)                & 0.377 (0.020)                & 0.063 (0.018)                \\
Hamming              & 6.955 (0.847)                & 0.070 (0.038)                & 0.320 (0.038)                & 0.061 (0.041)                \\ \bottomrule
\end{tabular}
\caption{Maximum likelihood estimates and their respective standard errors (in parentheses) of Kendall and Hamming R-GARCH(3,0) models fitted to the tennis rankings data. }\label{tab:fitted_complete_tennis}
\end{table}


\begin{figure}
    \centering
    \includegraphics[width=0.9\linewidth]{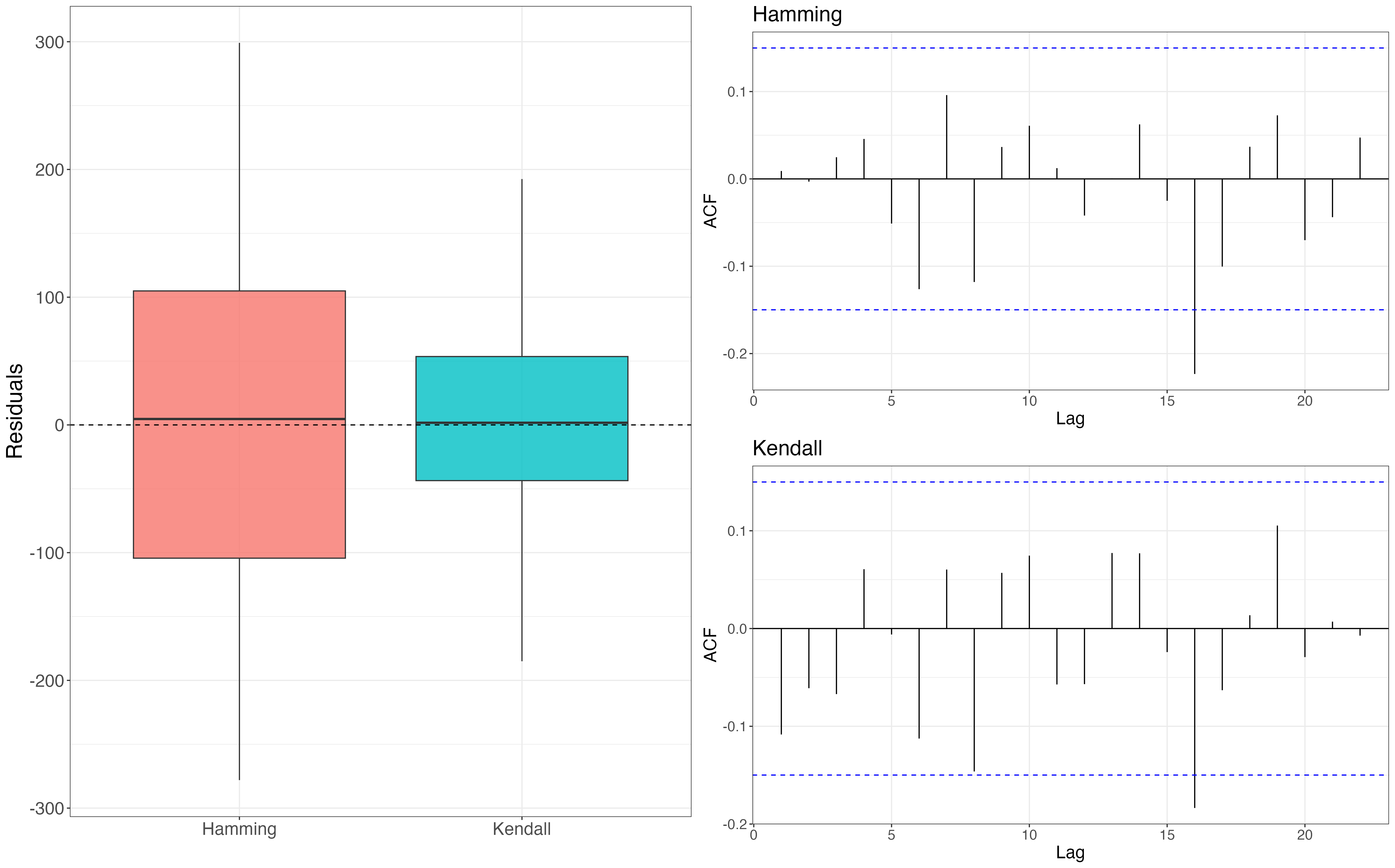}
    \caption{Residual analysis of the Kendall and Hamming R-GARCH(3,0) models fitted to the tennis rankings data.} \label{fig:residuals}
\end{figure}


\subsection{Incomplete Tennis Ranking Analysis}

In this subsection, we analyze the ten most competitive players from the 2015-2019 period. This classification is made based on the highest number of aggregate points per athlete, resulting in the set: \texttt{Novak Djokovic, Rafael Nadal, Roger Federer, Andy Murray, Dominic Thiem, Alexander Zverev, Marin Cilic, Kei Nishikori, Stan Wawrinka} and \texttt{Milos Raonic}. The filtered data that focuses on this set is a $172 \times 10$-dimensional matrix, which is represented in Figure \ref{fig:top10_ts}. The heatmap in this figure illustrates the ranking trajectories attributed to each player over time. Players are represented as rows and tile shades (from green to red) indicate their ranking each week. Tiles in grey color point out missing entries which happen for \texttt{Stan Wawrinka} and \texttt{Andy Murray} in June-August 2018.

\begin{figure}
    \centering
    \includegraphics[width = 0.7\linewidth]{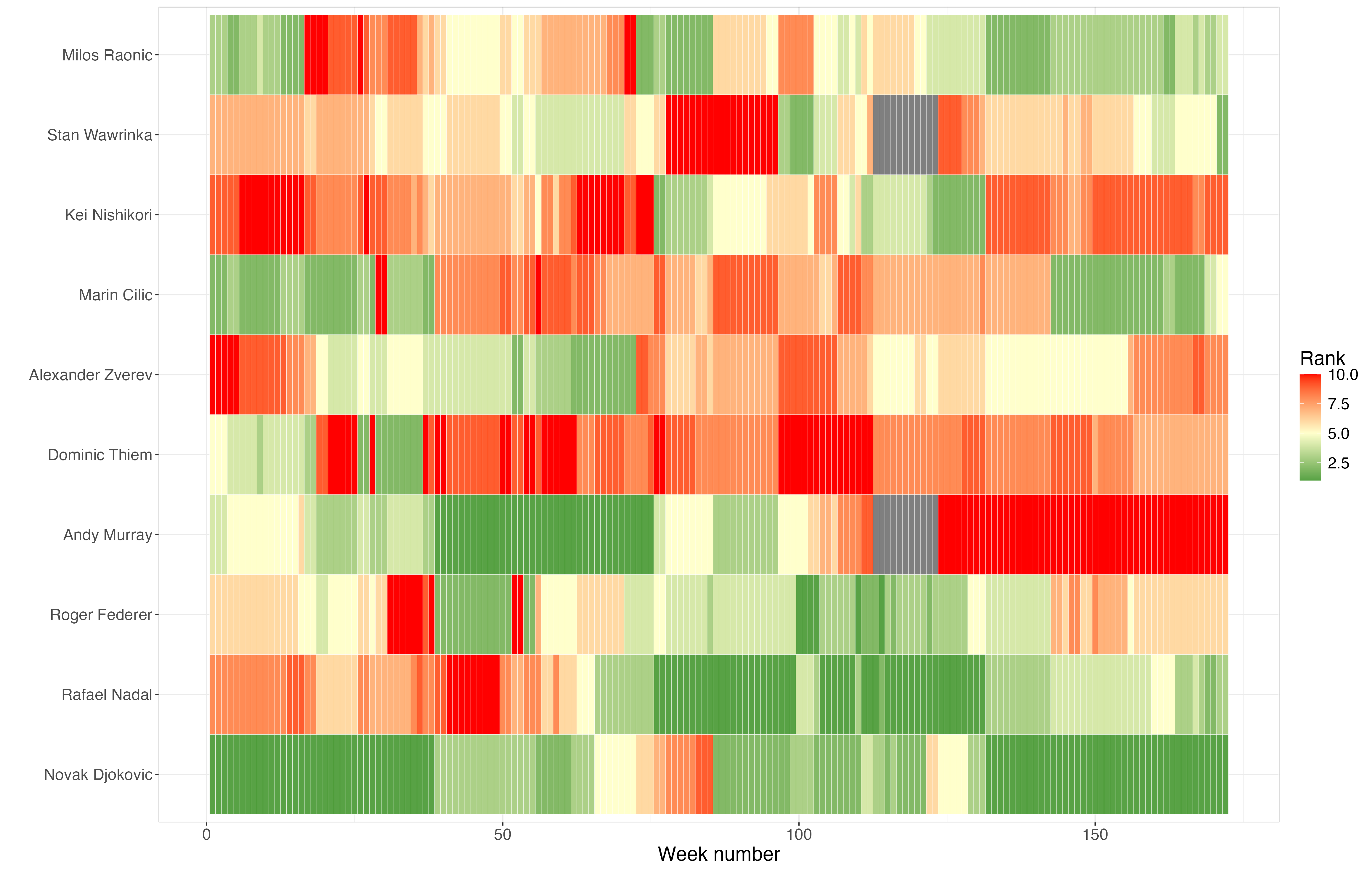}
    \caption{Rankings from top players (rows) at each week (columns). Each cell is colored according to the respective ranking position according to the legend key shown to the right.} \label{fig:top10_ts}
\end{figure}

The inferential strategy designed to handle incomplete rankings presented in Section \ref{sec:incomplete} is applied to this data. To reduce the computational burden of running the MCEM algorithm multiple times, the selection of $p$ and $q$ is made beforehand. This is done as follows. First, complete data is generated by sampling missing entries uniformly at random from the unranked set. Then, the Kendall and Hamming R-GARCH models are fitted with $p,q \in \{0, \ldots, 3\}$ using standard maximum likelihood estimation. This is repeated sometimes to mitigate sampling variation, and the results point order $(p,q)=(3,0)$ in the majority of times for both models. Parameter estimates and standard errors for the R-GARCH(3,0) models are then obtained with MCEM using $M=500$ and $\epsilon = 10^{-3}$. The results are displayed in Table \ref{tab:params_incomplete}. Models are compared via analysis of residuals, which is carried out as before. Although both successfully capture the data's autocorrelation, smaller residuals are obtained from the Kendall fit. The latter is then selected for a predictive task that is described in what follows.

\begin{figure}[!htb]
    \centering
    \begin{minipage}{0.4\textwidth}
        \centering
\begin{tabular}{@{}lll@{}}
\toprule
\multirow{2}{*}{} & \multicolumn{2}{l}{Estimate (SE)} \\ \cmidrule(l){2-3} 
                  & Kendall         & Hamming         \\ \midrule
$\phi_0$          & 3.700 (0.312)   & 3.424 (0.296)   \\
$\phi_1$          & 0.000 (0.035)   & 0.000 (0.034)   \\
$\phi_2$          & 0.198 (0.034)   & 0.199 (0.035)   \\
$\phi_3$          & 0.067 (0.031)   & 0.104 (0.032)   \\ \bottomrule
\end{tabular}
\captionof{table}{Kendall and Hamming R-GARCH(3,0) model parameter estimates and their standard errors (in parentheses) via MCEM for the top ten tennis players data.}\label{tab:params_incomplete}
    \end{minipage}\hfill
     \begin{minipage}{.55\textwidth}
        \centering
\includegraphics[width = .8\linewidth]{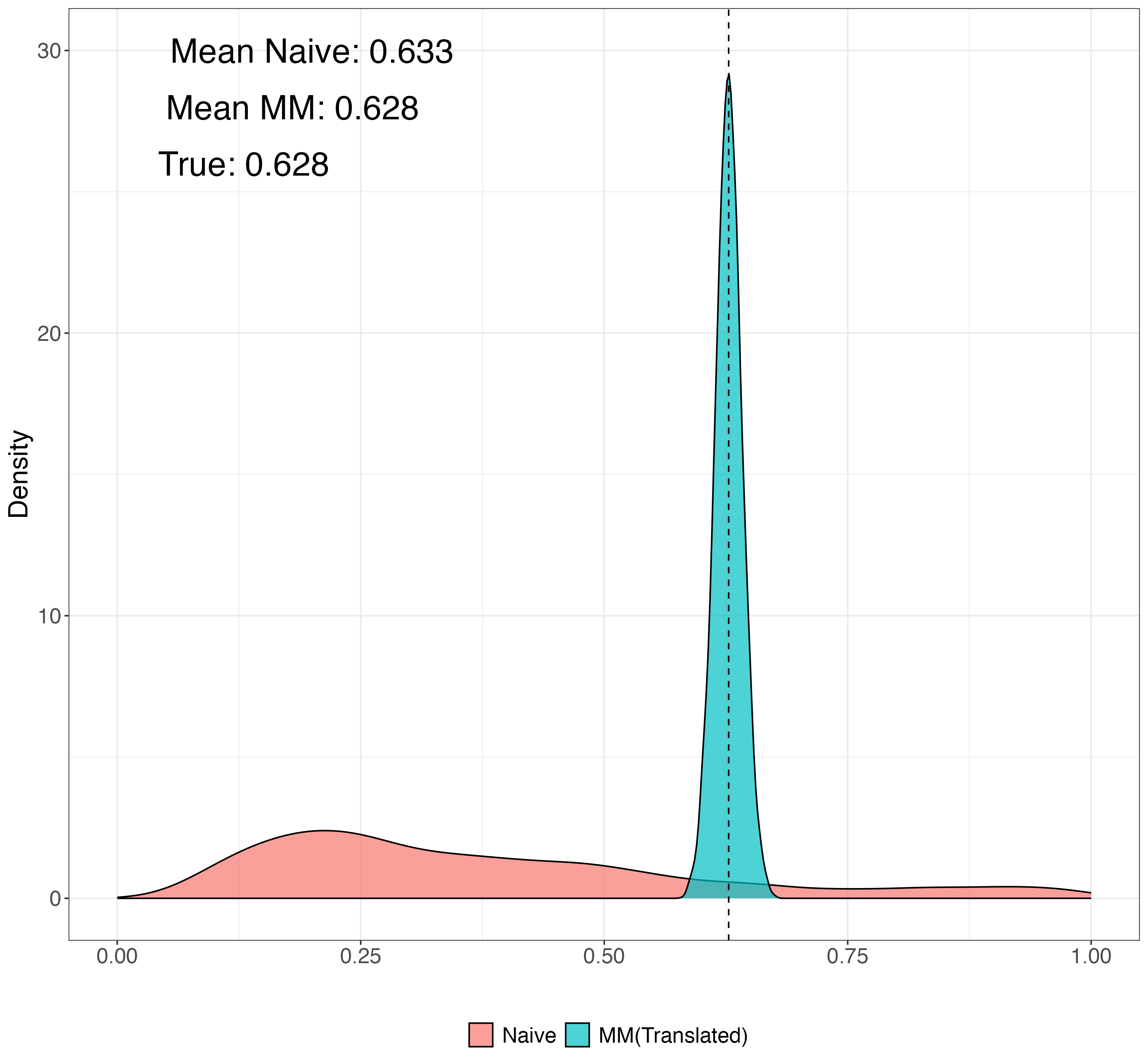}
             \captionof{figure}{Replicated importance sampling estimates of $\mbox{pr}_{t+1}$ obtained under naive and designed Mallows model importance densities. }\label{fig:IS}
    \end{minipage}
\end{figure}

One natural point of interest is to make predictions with the fitted model. Under the R-GARCH model, it is straightforward to obtain the one-step-ahead forecast for the average distance from the current rank, i.e. $E[d(\pi_{t+1}, \pi_t)|\mathcal F_{t}]$. Although this gives a sense of how close to $\pi_{t}$ the next ranking is expected to be, it does not say much about the permutation itself. For example, one could be interested in some quantity involving the conditional distribution of $\pi_{t+1}$ given $\pi_{t}$ such as: \textit{What is the probability that the top player in week $t$ maintains his position in the next week?} Mathematically, this question translates as computing $\mbox{pr}_{t+1}\equiv \displaystyle\sum_{\pi_{t+1} \in \mathcal{P}_k} I\{ \pi_{t+1}(1) = \pi_{t}(1)\} \frac{ \exp( -\theta_{t+1} d(\pi_{t+1}, \pi_t) )}{\Psi(\theta_{t+1})}$, where $I\{\cdot\}$ is as an indicator function. We replace $\theta_{t+1}$ by its estimate derived from $\widehat{\mu}_{t+1}$ via Algorithm \ref{alg:solve_theta}. However, generating all possible $\pi_{t+1}$ compatible with the condition of interest can be unfeasible depending on the value of $k$. To tackle this, an estimator of $\mbox{pr}_{t+1}$ is designed based on Importance Sampling (IS), which will be presented in what follows.

Consider the above question of interest and $t = 172$. The player ranked first at this time was \texttt{Novak Djokovic}, and we would like to know the probability that he retains the first rank in the next classification. Under IS, an estimator of $\mbox{pr}_{t+1}$ is constructed by sampling $\pi_{t+1}$ from some (importance) distribution, say $g(\cdot)$. The IS estimator based on $L$ samples from $G$, say $\{\pi^{(l)}\}_{l=1}^L$, is given by
\begin{eqnarray}\label{eq:is_estimator}
\widehat{\mbox{pr}}_{t+1} = \frac{1}{L} \sum_{l=1}^L h(\pi^{(l)}_{t+1}(1)) \frac{p(\pi_{t+1}^{(l)}|\mathcal F_t,\widehat{\theta}_{t+1})}{g(\pi^{(l)}_{t+1})},    
\end{eqnarray}
where $h(\pi^{(l)}_{t+1}(1))$ is a function of $\pi_{t+1}$; for our purpose, $h(\pi^{(l)}_{t+1}(1)) = I\{\pi^{(l)}_{t+1}(1) = \pi_{t}(1)\}$. Well-known features of IS are that the estimator is unbiased and its variability depends crucially on the importance density. To reduce the variance, we would like to simulate from a distribution that is proportional to $h(\pi^{(l)}_{t+1})p(\pi_{t+1}^{(l)}|\mathcal{F}_t,\widehat{\theta}_{t+1})$ such that $g(\cdot)$ is high when $h(\cdot) p(\cdot|\mathcal F_t,\widehat{\theta}_{t+1})$ is high. Evidently, the first step in this direction is to fix $I\{ \pi^{(l)}_{t+1}(1) = \pi_t(1)\} = 1$. Next, observe that $p(\pi_{t+1}|\mathcal F_t,\widehat{\theta}_{t+1})$ places high weight at rankings $\pi_{t+1}$ that make $d(\pi_{t+1}, \pi_t)$ low. Hence, a good importance distribution is a Mallows model with mode $\pi_t$. We can also make the spread parameter of this Mallows importance density similar to that of $p(\pi_{t+1}|\mathcal F_t,\widehat{\theta}_{t+1})$ by `translating' $\widehat{\theta}_{t+1}$ to the $k - k_{\mathcal{I}}$ setting, where $k_{\mathcal{I}}$ is the number of positions fixed under $I\{\cdot\}$. Based on these ideas, our $g(\cdot)$ choice is a Mallows distribution with mode given by the relabel of the unfixed $\pi_t$ and spread $\theta^* = \widehat{\theta}_{t+1} {k - k_{\mathcal{I}} \choose 2}/{k \choose 2}$ where ${k - k_{\mathcal{I}} \choose 2}/{k \choose 2}$ is the ratio of maximum disagreements under $k - k_{\mathcal{I}}$ and $k$. 

To demonstrate that this strategy is effective and achieved low IS variability, we provide its comparison to a naive $g(\cdot)$ choice. The latter consists of simply fixing $I\{\pi^{(l)}_{t+1}(1) = \pi_{t}(1)\} =1$ and sampling all the remaining items uniformly at random, hence $g(\pi) = 1/(k-1)!$. The alternative estimators of $\mbox{pr}_{t+1}$ are compared by gathering 10K replications of each, setting $L= 500$ in both. Results are displayed in Figure \ref{fig:IS} via histograms. Given that $k=10$, it is possible to compute $\mbox{pr}_{t+1}$ exactly for this application and compare the estimates with the exact probability. This is marked with a vertical dashed line and displayed on the top-left corner alongside the average $\widehat{\mbox{pr}}_{t+1}$ per method. As expected, the empirical mean of $\widehat{\mbox{pr}}_{t+1}$ is close to the exact probability under both estimation procedures, which follows from an IS property. However, the variance decrease is drastic under the well-designed importance density. Having considered a case where $\mbox{pr}_{t+1}$ is tractable allowed us to check and validate the proposed approach. The designed IS algorithm makes it possible to make predictions on the rankings space when $k$ is moderate or large.


\section{Conclusions}\label{sec:conclusion}

This paper presents a novel approach for modelling ranking time series data. Our formulation integrated two widespread methods within the rankings and time series literature: the Mallows model (MM) and the GARCH approach. This integration, denoted by R-GARCH models, was achieved via a reparameterisation of the Mallows distribution where autoregressive and feedback components were incorporated through its conditional mean. 
Theoretical and computational methods concerning the R-GARCH models were carefully developed. Stationarity and ergodicity properties of the process were established. Parameter estimation was addressed for both complete and incomplete rankings, using maximum likelihood estimation and a Monte Carlo EM algorithm, respectively. These inferential methods were validated via simulation studies and later applied to tennis ranking time series data. Specifically, we focused on the weekly rankings of professional male tennis players where model selection, diagnostic tools, and prediction were considered. To carry out predictions in the permutation space, an importance sampling (IS) technique was introduced. With the design of an effective importance density, our algorithm extended prediction on $\mathcal{P}_k$ to moderate or large $k$ number of individuals/items.

Finally, some points of future research that deserve attention are the following. Inference for the R-GARCH models under distances with non-analytical forms of either (or both) expected value and normalisation constant is a challenging point. This may require computationally intensive methods, particularly in root-solving algorithms and intractable inference procedures. In another vein, for applications where additional information is available, a natural goal is the incorporation of exogenous variables.

\spacingset{1} 

\end{document}